\documentclass[12pt]{article}
\usepackage{amsmath}
\usepackage{amsfonts}
\usepackage{amssymb}
\usepackage{amsthm}
\usepackage{cite}
\usepackage[unicode=true,pdfusetitle,
 bookmarks=true,bookmarksnumbered=false,bookmarksopen=false,
 breaklinks=true,pdfborder={0 0 0},backref=false,colorlinks=true]
 {hyperref}
\usepackage{mathrsfs}
\usepackage{subcaption}
\usepackage{cancel}

\usepackage{color}
\newcommand{\bc}{\color{black}}
\newcommand{\bcc}{\color{black}}

\newlength{\dinwidth}
\newlength{\dinmargin}

\newif\ifshowdetails
\showdetailstrue

\hypersetup{
 linkcolor=blue,citecolor=blue, urlcolor=blue}

\newcounter{todo}

\newcommand{\todobox}[1]{
  \textcolor{blue}{
    \fbox{\parbox{0.9\textwidth}{#1}}%
  }
}

\newcommand{\todonotetag}{TODO\thetodo}
\newcommand{\todonote}[1]{%
\stepcounter{todo}%
{\let\thefootnote\todonotetag
\footnote{\todobox{#1}}%
}
}

\definecolor{detailsgray}{gray}{0.3}

\ifshowdetails
\newcommand{\detail}[1]{%
{\color{detailsgray}$\blacktriangleright${#1}$\blacktriangleleft$}%
}

\newcommand{\detailspar}[1]{
\par \noindent {\color{detailsgray} $\blacktriangleright$ \textit{#1} $\blacktriangleleft$ } \par
}

\newenvironment{details}{\par \noindent \begingroup\itshape\color{detailsgray} $\blacktriangleright$ \hrulefill%
\par \noindent \ignorespaces}{\par\noindent \hrulefill $\blacktriangleleft$  \endgroup \par \noindent}

\else
\newcommand{\detail}[1]{} 
\newcommand{\detailspar}[1]{} 
\usepackage{environ}
\NewEnviron{hide}{}

\fi

\usepackage{tikz}
\usetikzlibrary{trees}
\usetikzlibrary{decorations.pathmorphing}
\usetikzlibrary{decorations.markings}

\newcommand{\thetae}{\eps}

\newcommand{\alt}{\tilde{\al}}
\newcommand{\sect}{ \mrm{In}(\mfa)}
\newcommand{\sym}{\mathbf{\sigma}}
\newcommand{\Ss}{S}
\newcommand{\Rr}{R}

\newcommand{\Aut}{\mrm{Aut}}

\newcommand{\vvp}{\pmb{\vp}}
\newcommand{\vJ}{\pmb{J} }
\newcommand{\vS}{\pmb{S} }
\newcommand{\vL}{\pmb{L}}
\newcommand{\pol}{\pmb{\epsilon}_{\la}}

\newcommand{\vv}{\pmb{v}}
\newcommand{\hGa}{\hat{\Ga}}
\newcommand{\y}{\!\!\!}
\newcommand{\mrm}{\mathrm}
\newcommand{\vP}{\mathbf{P}}

\newcommand{\vA}{\mathbf{A}}

\newcommand{\vf}{\mathbf{f}}

\newcommand{\vg}{\mathbf{g}}

\renewcommand{\mathbf}{\boldsymbol}

\newcommand{\Span}{\mathrm{Span}}

\newcommand{\mcF}{\mathcal F}

\newcommand{\mcL}{\mathcal L}

\newcommand{\mcS}{\mathcal S}

\newcommand{\vx}{\boldsymbol{x}}

\newcommand{\ti}{\tilde}

\newcommand{\vac}{{\bc \mrm{vac}}}
\newcommand{\Om}{\Omega}
\newcommand{\ga}{\gamma}

\newcommand{\ka}{\kappa}

\newcommand{\vk}{\boldsymbol{k}}

\newcommand{\be}{\beta}

\newcommand{\pa}{\partial}

\newcommand{\ov}{\overline}
\newcommand{\vp}{\varphi}
\newcommand{\mfh}{\mathfrak{h}}

\newcommand{\eps}{\varepsilon}
\newcommand{\de}{\delta}

\newcommand{\De}{\Delta}

\newcommand{\pho}{\mathrm{ph}}

\newcommand{\nin}{\noindent}
\newcommand{\si}{\sigma}
\newcommand{\ph}{\phantom}
\newcommand{\h}{\fr{1}{2}}
\newcommand{\nat}{\mathbb{N}}
\newcommand{\hil}{\mathcal{H}}
\newcommand{\om}{\omega}
\newcommand{\mfa}{\mathfrak{A}}

\newcommand{\fr}[2]{\frac{#1}{#2}}
\newcommand{\al}{\alpha}
\newcommand{\real}{\mathbb{R}}
\newcommand{\complex}{\mathbb{C}}
\newcommand{\la}{\lambda}
\newcommand{\non}{\nonumber}
\newcommand{\Ga}{\Gamma}

\newcommand{\lan}{\langle}
\newcommand{\ran}{\rangle}

\def\qed{$\Box$\medskip}

\newtheorem{theoreme}{Theorem } [section]
\newtheorem{proposition}[theoreme]{Proposition}
\newtheorem{lemma}[theoreme]{Lemma}
\newtheorem{definition}[theoreme]{Definition}
\newtheorem{corollary}[theoreme]{Corollary}
\newtheorem{remark}[theoreme]{Remark}
\newtheorem{example}[theoreme]{Example}
\newtheorem{criterion}[theoreme]{Criterion}
\newtheorem{conjecture}{Conjecture}
\newtheorem{assumption}{Assumption}

\newcommand{\tr}{\mrm{tr}}

\newcommand{\bea}{\begin{assumption}}
	\newcommand{\eea}{\end{assumption}}
\newcommand{\beco}{\begin{conjecture} }
	\newcommand{\eeco}{\end{conjecture} }
\newcommand{\beq}{\begin{equation}}
	\newcommand{\eeq}{\end{equation}}
\newcommand{\beqa}{\begin{eqnarray}}
	\newcommand{\eeqa}{\end{eqnarray}}
\newcommand{\ben}{\begin{arabicenumerate}}
	\newcommand{\een}{\end{arabicenumerate}}
\newcommand{\bex}{\begin{example}}
	\newcommand{\eex}{\end{example}}
\newcommand{\ber}{\begin{remark}}
	\newcommand{\eer}{\end{remark}}
\newcommand{\bec}{\begin{corollary}}
	\newcommand{\eec}{\end{corollary}}
\newcommand{\bep}{\begin{proposition}}
	\newcommand{\eep}{\end{proposition}}
\newcommand{\becr}{\begin{criterion}}
	\newcommand{\eecr}{\end{criterion}}

\def\bel{\begin{lemma}}
	\def\eel{\end{lemma}}
\def\bet{\begin{theoreme}}
	\def\eet{\end{theoreme}}
\def\bed{\begin{definition}}
	\def\eed{\end{definition}}

\begin{document}

\title{Relative normalizers of automorphism groups, infravacua and the problem of velocity superselection in  QED}

\author{Daniela Cadamuro\footnote{E-mail: daniela.cadamuro@itp.uni-leipzig.de} \ and Wojciech Dybalski\footnote{E-mail: dybalski@ma.tum.de}  \\\\
Zentrum Mathematik, Technische Universit\"at M\"unchen,\\
D-85747 Garching, Germany   }

\date{}

\maketitle

\begin{abstract} 
We advance superselection theory of pure states on a $C^*$-algebra $\mfa$  outside of
the conventional (DHR) setting. First, we canonically define conjugate and second conjugate
classes of such states with respect to a given  reference state $\om_\vac$ and background  $a \in \Aut(\mfa)$.
Next, for some subgroups
 $R  \varsubsetneq S  \varsubsetneq G\subset  \Aut(\mfa)$ we study the  family $\{\, \om_\vac\circ s\,|\,  s\in S \}$
 of \emph{infrared singular states} whose superselection sectors may be disjoint for different $s$. 
We show that their  conjugate and second conjugate classes always coincide provided that $R$ leaves the sector of $\om_\vac$
invariant and  $a$  belongs to the \emph{relative normalizer} $N_G(R,S):=\{\, g\in G\,|\, g\cdot S\cdot g^{-1}\subset R\,\}$.
We study the basic properties of this apparently new group theoretic concept and show that the Kraus-Polley-Reents infravacuum automorphisms
belong to {\bcc the} relative normalizers of  the automorphism group of a suitable CCR algebra. Following up on
this observation we show that the problem of velocity superselection in non-relativistic QED  disappears at
the level of conjugate and second conjugate classes, if they are computed with respect to an  infravacuum background.
We also demonstrate that for more regular backgrounds such merging effect does not occur.
\end{abstract}

\section{Introduction}
\setcounter{equation}{0}

The infrared problem is a maze of difficulties in spectral, scattering and superselection theory of quantum systems which can be traced back to the presence of massless particles
and long-range forces.  This topic enjoys currently some revival in the high-energy physics community triggered by 
 a proposal of  Hawking, Perry and Strominger to solve  the black hole information paradox using the  infrared degrees of freedom
of the gravitational field \cite{HPS16}.  This proposal  relies on an implicit assumption that these infrared degrees of freedom  can be  measured,
in particular have sufficiently mild fluctuations. In the case of QED this assumption {\bcc (cf. \cite[formula (2b)]{Bu86})} lies at the heart of the infrared problem. Its consequences include several 
closely related pathologies which can collectively be called the \emph{infraparticle problem}:  superselection of the electron's velocity 
\cite{Fr73, CF07, CFP09, KM14, DT11}, {\bc the} absence of the sharp mass  of the electron \cite{Bu86, HH08}, and infrared divergencies in the Dyson scattering matrix \cite{We}.   In view of the above, a natural approach to cure the infrared problems
is to immerse the system in a low-energetic but highly fluctuating radiation which blurs the infrared degrees of freedom. A concrete example of such an
\emph{infravacuum} was given by Kraus, Polley and Reents already four decades ago in the context of {\bc the quantized electromagnetic field coupled to an external current} \cite{Re74, KPR77, Kr82}.
The initial success  of the approach was a well-defined Dyson scattering matrix in this exactly solvable situation. However, the  problem of velocity superselection and sharp mass of the electron were not treated in these works, mostly due to the absence (back then)  of mathematically tractable, translation-invariant models of QED and limited understanding of superselection theory in the presence of long-range forces.   As the intervening decades witnessed  progress in mathematical control of such systems in the infrared regime
(see, e.g., \cite{BFS98.2,  BCFFS13,  Pi03,  Pi05, CF07, CFP07, CFP09, HH08, Hi00, KM14, LMS07, DP13.1, DP17.2}) and substantial advances on the side of general superselection theory \cite{Bu82, BR14} 
there is now every reason to  re-initiate the infravacuum program. 

In the present paper we point out that the infravacua exemplify a general group-theoretic concept which apparently escaped attention so far:  Let $R  \varsubsetneq S  \varsubsetneq G$
be subgroups of a group $G$ and let us call 
\beqa
N_{G}(R,S):=\{\, g\in G\,|\, g\cdot S\cdot g^{-1}\subset R\,\}
\eeqa
the \emph{relative normalizer} of the pair of
subgroups $(R,S)$.  Due to the tension between the inclusion $\Rr \varsubsetneq \Ss$ and the opposite inclusion $g\cdot \Ss \cdot g^{-1}\subset \Rr$ in the definition of $N_{G}(R,S)$, one can immediately exclude non-empty  relative normalizers   in many situations. 
In particular, relative normalizers are empty for any subgroups of an abelian group $G$.  
The same is true for  finite groups (abelian or not), since elements of the relative normalizer would provide bijections between sets
of different cardinality. Similarly, closed, connected subgroups of finite-dimensional Lie groups have  empty relative normalizers, as 
their elements would provide continuous bijections between sets of different dimension (cf. \cite[Proposition~5.1, {\bc Theorem 5.6}]{Sch16}).  
Finding non-trivial examples of relative normalizers is also not difficult, for example in the context of infinite discrete groups 
(see Subsection~\ref{relative-normalizer-subsection}). However, our main interest in this paper is in a certain concrete class of relative normalizers,
which is relevant for infrared problems:
Let $\mcL$ be the photon single particle space with a suitable IR regularisation. We show that the inhomogeneous symplectic group $\mrm{ISp}(\mcL)$
admits  relative normalizers which contain the Kraus-Polley-Reents symplectic maps $T$ (cf.~Subsection~\ref{semi-direct}, Theorem~\ref{infravacuum-theorem} and Definition~\ref{infravacuum}).  These relative normalizers can be lifted to the group of automorphisms $\Aut(\mfa)$ of the corresponding CCR algebra~$\mfa$.
Indeed, with the help of the automorphism $\al$ which maps elements of $\mrm{ISp}(\mcL)$ into the corresponding 
Bogoliubov transformations in  $\Aut(\mfa)$  we  show that $\al_T$ belongs to the relative normalizer
in any subgroup $G$ of $\Aut(\mfa)$ containing $\al_{ \mrm{ISp}(\mcL) }$. In this context $R$ and $S$ are the subgroups of automorphisms $\al_{\mcL^*_\mrm{R}}$ and $\al_{\mcL^*_\mrm{S}}$ of $G$, where $\mcL^*_\mrm{R}$ is a subspace of the single photon space without IR regularisation and $\mcL^*_\mrm{S}$ is the sum of $\mcL^*_\mrm{R}$ with the set of infrared singular dressing transformations of the bare electron for different momenta.


Our interest in relative normalizers derives from their relevance for superselection theory of the $C^*$-algebra $\mfa$. Let $P_{\mfa}$ be the set of pure states
and $X:=P_{\mfa}/ \mrm{In}(\mfa)$ the set of sectors.  {\bc The latter are} the orbits of states under the action of inner automorphisms of $\mfa$ 
(cf.  \cite[Definition 4.1]{BR14})  and for a state $\om\in P_{\mfa}$ such orbit will be denoted by $[\om]_{\sect}$.
Now we consider the natural  right action of a subgroup $G\subset \Aut(\mfa)$ on $X$, fix a reference (`vacuum') sector $x_0=[\om_{{\bc \mrm{vac}}}]_{\sect}$ and a `background
automorphism' $a\in G$ {\bc corresponding to background radiation}. With this data, for any sector $x\in X$ we define the \emph{conjugate class} in a way which is suggested by the 
DHR superselection theory of simple charges \cite{Ha}:
\beqa
\ov{[x]}^{a}:=\{\,x_0\cdot a\cdot  g^{-1}\,|\, g\in G^{a}_{x, x_0}\,\}, \textrm{ where }  G_{x, x_0}^{a}:=\{\, g\in G\,|\, x=x_0\cdot a\cdot g\,\}.
\eeqa
By iterating this procedure, we also define the \emph{second conjugate class} $\ov{\ov{[x]}}^{a}$. Now let  $R  \varsubsetneq S  \varsubsetneq G$ be subgroups as above 
and let us consider the family  $\{\, x_0\cdot s\,|\, s\in S\, \}$ of \emph{singular sectors}. Our main general result is that
\beqa
\ov{[x_0\cdot s_1]}^{a}=\ov{[x_0\cdot s_2]}^{a} \textrm{ and } \ov{\ov{[x_0\cdot s_1]}}^{a}=  \ov{\ov{[x_0\cdot s_2]}}^{a}  \label{sector-identities}
\eeqa
hold for all $s_1, s_2\in S$, provided that $x_0\cdot r=x_0$ for all $r\in R$   and $a\in  N_{G}(R,S)$. These equalities are non-trivial as long as $x_0\cdot s_1\neq
x_0\cdot s_2$ for some $s_1\neq s_2$  and we give such examples below.  We also demonstrate that the identities (\ref{sector-identities}) may fail, if the
background $a$ is not in  $N_{G}(R,S)$. We stress that the above considerations are purely group-theoretical, and we referred to $\Aut(\mfa)$ only 
for concreteness {\bc and motivation}.

If a sector $x\in X$ belongs to the $G$-orbit of $x_0$, it is easy to see that  $x \in \ov{\ov{[x]}}^{a}$. More than that, if $a\in N_{G}(R,S)$ then also
$x\cdot s \in  \ov{\ov{[x]}}^{a}$ for all $s\in S$ (cf. Theorem~\ref{disjointness-proposition-alg}~(c)).  
In the context of the physical example below,
the second conjugate class 
with respect to an  \emph{infravacuum  background}  absorbs a multitude of distinct sectors  $x\cdot s$, $s\in S$, which differ only by  physically irrelevant  
\emph{soft-photon {\bc clouds}}. 
Thus the  second conjugate class  appears to be a natural generalisation of the concept of a sector for theories with infrared problems. It is similar -- in intention --
to the \emph{charge classes} introduced in \cite{Bu82, BR14}, but it does not rely on  locality.  Also, in contrast to earlier discussions of superselection theory
with respect to infravacuum backgrounds \cite{Bu82, Ku, Ku98.1}{\bc ,} we do not use any variant of the DHR criterion. 
This facilitates  applications of our general results  to  non-relativistic, interacting models of QED,  as we now summarize.

The Hilbert space of the model {\bc we consider} is $\hil=L^2(\real^3)\otimes \mcF_{\pho}$, where $L^2(\real^3)$ carries the electron degrees of freedom and   $\mcF_{\pho}$ is the Fock space of the physical photon states.  The Hamiltonian has the standard form \cite{Sp}
\beqa
H:= \frac{1}{2}(-i\nabla_{\pmb{x}} + \ti{\alpha}^{1/2} \pmb{A}(\pmb{x}))^2 + H_{\pho},  \label{intro-Hamiltonian}
\eeqa
where $\ti{\al}>0$ is the coupling constant, $\pmb{x}$ is the position of the electron, $\pmb{A}$ is the electromagnetic potential in the Coulomb gauge with fixed ultraviolet regularization and 
$H_{\pho}$ is the free photon Hamiltonian. This Hamiltonian is translation-invariant, that is, it commutes with the total momentum operators $\mathbf{P}=-i\nabla_{\pmb{x}}+\mathbf{P}_{\pho}$, where $\mathbf{P}_{\pho}$ is the  free-photon momentum. Consequently, we can decompose it into the fiber Hamiltonians $H_{\pmb{P}}$ at fixed momentum
\begin{equation}
 H =\Pi^* \bigg(\int^{\oplus} H_{\pmb{P}} \; d^3 \pmb{P}\bigg)\Pi, \label{fiber-hamiltonians}
\end{equation} 
where $\Pi$ is a certain unitary identification. 
The Hamiltonians  $H_{\pmb{P}}$, precisely defined in Section~\ref{sec:model} below, are self-adjoint operators acting on  the fiber Fock space which we denote by  $\mcF$. It is one manifestation of the infraparticle problem that $H_{\pmb{P}}$ do not have ground states for $\vP\neq 0$, at least for small $\ti\al$  and for $\vP$ in some ball $\mcS$ around zero \cite{HH08, CFP09}.  By introducing an infrared cut-off 
$\sigma>0$ in the  interaction term in (\ref{intro-Hamiltonian}), we obtain the Hamiltonian $H_{\sigma}$ and the corresponding fiber Hamiltonians 
$H_{\pmb{P},\sigma}$, which do have the (normalized) ground states $\Psi_{\pmb{P},\sigma}$ in the same region of parameters $\alt, \vP$. {\bc Although} these vectors tend weakly to zero as $\sigma\to 0$ \cite{CFP09},
one obtains well-defined states on a certain CCR algebra $\mfa$
\footnote{ This CCR-algebra is  defined in Subsections~\ref{algebraic-preliminaries}, \ref{relative-norm-subsection} below and   
we give a proof of velocity superselection in our setting in Section~\ref{proof-of-velocity-superselection}. In earlier works \cite{Fr73, CF07, CFP09, KM14} 
 slightly different algebras are 
used.  For example, in \cite{CF07} $\mfa=\ov{\bigcup_{\si>0} B(\mcF_{\si})}^{\|\,\cdot\,\|}$, where $\mcF_{\si}\subset \mcF$ is the subspace of functions vanishing in a ball of radius $\si$ around zero (in any variable).  } 
\beqa
\om_{\pmb{P}}(A){\bc :=}\lim_{\sigma\to 0}\lan \Psi_{\pmb{P},\sigma},\pi_{{\bc \mrm{vac}}}(A)  \Psi_{\pmb{P}, \sigma}\ran, \quad A\in\mfa, \label{original-Psi-0}
\eeqa
where $\pi_{{\bc \mrm{vac}}}$  is the Fock {\bc vacuum} representation.
These states can be interpreted as plane-wave configurations of the electron moving with momentum $\vP$.
It is well known that in (\ref{intro-Hamiltonian}), and in similar models of non-relativistic QED,  the GNS representations $\pi_{\pmb{P}}$
of the states $\om_{\pmb{P}}$  are disjoint for different values of $\pmb{P}\in \mcS$  \cite{Fr73, CF07, CFP09, KM14}. This 
mathematical formulation of velocity superselection was first introduced by Fr\"ohlich in \cite{Fr73}.

Given the $C^*$-algebra $\mfa$ and the family of distinct sectors $[\om_{\vP}]_{\sect}\neq [\om_{\vP'}]_{\sect}$ for $\vP\neq \vP'$ we can compute their conjugate and second conjugate classes.
We choose the Fock vacuum $[\om_{\bc \mrm{vac}}]_{\sect}$ as a reference sector and a Kraus-Polley-Reents infravacuum automorphism $\al_{T}$ as a background. By exhibiting
a concrete relative normalizer including $\al_T$, we obtain from (\ref{sector-identities}) that
\beqa
\ov{[ [ \om_{\vP} ]_{\sect} ]}^{\al_T}=\ov{[ [ \om_{\vP'}  ]_{\sect}  ]}^{\al_T} \textrm{ and } \ov{\ov{[ [ \om_{\vP}]_{\sect} ]}}^{\al_T}=  \ov{\ov{[ [ \om_{\vP'}]_{\sect} ]}}^{\al_T},  \label{sector-identities-one}
\eeqa
for all $\vP, \vP'\in \mcS$. Thus  velocity superselection disappears at the level of conjugate and second conjugate
classes. Furthermore, if $\al_T$ in (\ref{sector-identities-one}) is replaced with some regular background (e.g., the identity automorphism), the velocity superselection
persists, at least for $\vP=0$, $\vP'\neq 0$. These results, stated precisely in Theorem~\ref{disjointness-proposition-nr-qed} below, illustrate the utility of our general theory in a concrete model.

Our paper is organized as follows: In Section~\ref{General-theory} we introduce the concepts of relative
normalizers and (second) conjugate classes, and prove their general properties. We also explain the 
relevance of these group-theoretic considerations to  superselection theory of $C^*$-algebras. Section~\ref{CCR-theory} concerns relative
normalizers in the inhomogeneous symplectic group and in the automorphism group of the corresponding CCR algebra.  These 
results are applied to the problem of velocity superselection in non-relativistic QED in  Section~\ref{QED-section}.
Section~\ref{Map-T} covers the definition and basic properties of {\bc the} Kraus-Polley-Reents infravacua and Section~\ref{proof-of-velocity-superselection}
gives a novel proof of velocity superselection.

\vspace{0.2cm}

\nin\textbf{Acknowledgements:} We would like to thank Henning Bostelmann, Detlev Buchholz, Fabio Ciolli, Maximilian Duell, Simon~Ruijsenaars and Yoh~Tanimoto for discussions
concerning the infravacuum representations. Thanks are also due to J\"urg Fr\"ohlich and Alessandro Pizzo for valuable remarks on non-relativistic
QED.  This work was  supported by the DFG within the Emmy Noether grant DY107/2-1.

\section{  Relative normalizers and conjugate classes }  \label{General-theory} 
\setcounter{equation}{0}

\subsection{Relative  normalizers}\label{relative-normalizer-subsection}
Let $H\subset G$ be a subgroup of a group $G$. Recall that the \emph{normalizer} of $H$ in $G$
is defined as
\beqa
N_{G}(H):=\{\, g\in G\,|\, g\cdot H\cdot g^{-1}= H\,\}
\eeqa
and it is the largest  subgroup of $G$ in which $H$ is normal. Also, we have the obvious relation:
\beqa
H\cdot N_{G}(H)\cdot H  =  N_{G}(H). \label{group-inclusions}
\eeqa
We generalize this concept as follows:
\bed Let $\Rr\subset \Ss\subset G$ be two subgroups of $G$. Then the \emph{relative normalizer} of the pair $(\Rr, \Ss)$ in $G$ is
defined as
\beqa
N_{G}(\Rr,\Ss):=\{\, g\in G\,|\, g\cdot S\cdot g^{-1}\subset R\,\}. \label{relative-normalizer}
\eeqa
\eed
Clearly, $N_{G}(\Rr,\Ss)$ is a semigroup, i.e., $N_{G}(\Rr, \Ss)\cdot N_{G}(\Rr, \Ss)\subset N_{G}(\Rr, \Ss)$, and  similarly to  
(\ref{group-inclusions}), we have
\beqa
\Rr\cdot N_{G}(\Rr,\Ss)\cdot \Ss   =    N_G(\Rr,\Ss).
\eeqa
If  $N_{G}(\Rr, \Ss)$ is a group then $\Rr=\Ss$ and $N_{G}(\Ss, \Ss)=N_{G}(S)$. 
On the other hand, $R=S$ only implies $N_{G}(\Ss, \Ss)\supset N_{G}(S)$. The inclusion
is proper and $N_{G}(\Ss, \Ss)$ fails to be a group  exactly if  $N_G(\ti R, S)$ is non-empty 
for some subgroup $\ti R  \varsubsetneq S$.
Furthermore,  for $\Ss=G$ the condition $N_{G}(\Rr, \Ss)\neq \emptyset$  implies $\Rr=G$ and $N_{G}(\Rr, \Ss)=G$. 
Leaving aside the case $N_{G}(\Ss, \Ss) \varsupsetneq N_{G}(S)$, we will say that a relative normalizer  $N_{G}(\Rr, \Ss)$  is non-trivial
if  $\Rr  \varsubsetneq \Ss \varsubsetneq G$.

We note for future reference that relative normalizers behave naturally under group homomorphisms.
\bel\label{homomorphism-lemma} Let  $\vp: G \to H$ be a  group homomorphism and $R\subset S\subset G$ be subgroups of $G$.   
If $g\in N_{G}(R,S)$ then $\vp(g)\in N_{\ti{H}}(\vp(R), \vp(S) )$, where $\ti H\subset H$ is any subgroup containing $\vp(G)$.
\eel
\proof By assumption, $g\cdot s\cdot g^{-1}\in R$ for any $s\in S$. Hence $\vp(g)\cdot \vp(s)\cdot \vp(g)^{-1}=\vp(   g\cdot s\cdot g^{-1}    )\subset \vp(R)$. \qed

The question of existence of non-trivial relative normalizers can readily be  settled by the following  example\footnote{We thank the anonymous referee for pointing out this example.}:
Let $\mathfrak{S}_{\mathbb{Z}}$ be the permutation group of integers. Let $\mathfrak{S}_{\nat}$, resp. $\mathfrak{S}_{\nat_0}$ be the subgroups
of $\mathfrak{S}_{\mathbb{Z}}$ consisting of permutations which act trivially on $\mathbb{Z}\backslash \nat$, resp. $\mathbb{Z}\backslash \nat_0$.  
Consider the inclusions
$ \mathfrak{S}_{\nat}\subset \mathfrak{S}_{\nat_0}\subset \mathfrak{S}_{\mathbb{Z}}.$
Let $a\in \mathfrak{S}_{\mathbb{Z}}$ be a permutation which shifts each element by one
to the right, i.e., $a(j)=j+1$, $j\in \mathbb{Z}$.
Then it is easy to see that $a\in N_{\mathfrak{S}_{\mathbb{Z}}}(\mathfrak{S}_{\nat},  \mathfrak{S}_{\nat_0})$.  One can find similar
examples in the context of diffeomorphism groups of intervals in $\real$  and unitary groups of Hilbert spaces.

\subsection{Group actions, orbits and conjugate classes}

Consider a group action of $G$ on a set $X$ which we denote  $ X\times G\ni (x,g)\mapsto  x\cdot g$.
(For future applications it is convenient to use the right action notation).
For any $x\in X$ we write
\beqa
G_{x}:=\{\, g\in G \,|\, x\cdot g=x\,\}
\eeqa
for the stabilizer subgroup of $x$. Furthermore, for any subgroup $H\subset G$ and $x\in X$ we denote the resulting orbit by
\beqa
[x]_{H}:=\{\,  x\cdot h\,|\, h\in H\,\}. 
\eeqa

Motivated by the DHR superselection theory,  we define for any $x\in X$ its conjugate and second conjugate class with respect to a certain reference point in $X$.
This is a counterpart of the inverse operation in $G$.
\bed\label{conjugation-def} Fix  reference elements $x_0\in X$ and $a\in G$.  For any $x\in X$ define the set 
$G_{x, x_0}^{a}:=\{\, g\in G\,|\, x=x_0\cdot a\cdot g\,\}$ and write
\beqa
\ov{[x]}^{a}:=\{\,x_0\cdot a\cdot  g^{-1}\,|\, g\in G^{a}_{x, x_0}\,\}, \quad \ov{\ov{[x]}}^{a}:=\{\, x_0\cdot a\cdot  (g')^{-1} \,|\,  g' \in G^{a}_{y, x_0}, \, y\in \ov{[x]}^{a} \, \}. \label{first-second-conjugate}
\eeqa
We call $\ov{[x]}^{a}$, (resp. $\ov{\ov{[x]}}^{a}$) the conjugate (resp. second conjugate) class of $x$ with respect to $(x_0,a)$.
\eed
\nin We note that $G_{x,x_0}^{a}$ is non-empty only if $x\in [x_0]_{G}$. Using this we obtain a simpler characterisation of conjugate {\bc classes} in terms of orbits.
In particular, it is clear from the following lemma that the conjugate classes do not change if $a$ is replaced with $g_0 \cdot a$ for some $g_0\in G_{x_0}$. 
\bel\label{conjugation-lemma} Let $x_0$, $a$  be as in Definition~\ref{conjugation-def} and 
  suppose that $x=x_0\cdot g_x$ for some $g_x\in G$. Then
  \beqa
\ov{[x]}^{a}= [x_0\cdot a\cdot g_x^{-1}\cdot a]_{ a^{-1}\cdot G_{x_0} \cdot a}, \quad 
   \ov{\ov{[x]}}^{a}=  [x_0 \cdot g_x]_{ a^{-1}\cdot G_{x_0}\cdot a },  
\label{conjugate-lemma}
\eeqa
where $G_{x_0}$ is the stabiliser group of $x_0$ and $ a^{-1}\cdot G_{x_0}\cdot a=G_{x_0\cdot a}$ is the stabilizer group of $x_0\cdot a$. 
\eel
\proof Since $G^a_{x, x_0}:=\{\, g\in G\,|\,x_0\cdot g_x= x_0\cdot a\cdot g\,\}$, 
we have  $G^a_{x, x_0}=a^{-1}\cdot G_{x_0}\cdot g_x$. Thus definition (\ref{first-second-conjugate}) gives 
\beqa
\ov{[x]}^{a}=x_0\cdot a\cdot g_x^{-1} \cdot G_{x_0} \cdot a \label{first-class-intermediate}
\eeqa
which is the  first formula in (\ref{conjugate-lemma}). To show the second formula in (\ref{conjugate-lemma}) we iterate
this argument, that is, we replace $g_x$ on the r.h.s. of (\ref{first-class-intermediate}) with $a\cdot g_x^{-1} \cdot G_{x_0} \cdot a$.
This gives
\beqa
\ov{\ov{[x]}}^{a}=x_0\cdot a\cdot (a\cdot g_x^{-1} \cdot G_{x_0} \cdot a)^{-1} \cdot G_{x_0} \cdot a=x_0\cdot g_x \cdot a^{-1}  \cdot G_{x_0} \cdot a,
\eeqa
which concludes the proof. \qed 
\begin{remark} It is clear from the proof of Lemma~\ref{conjugation-lemma} that any odd (resp. even)
number of conjugations, defined by iterating (\ref{first-second-conjugate}),  will reproduce the first (resp. the second) set in (\ref{conjugate-lemma}). Thus there
is no need to go beyond the second conjugate {\bc class}.
\end{remark}

In the following proposition we find $(x_0, a)$ from Definition~\ref{conjugation-def} for
which distinct points from  $[x_0]_G$ give rise to distinct conjugate and second conjugate classes.
\bep\label{disjointness-proposition} Let $(x_0,a)$ be as in  Definition~\ref{conjugation-def} 
 and suppose that $a^{-1}\cdot G_{x_0} \cdot a\subset  G_{x_0}$. Then, for all $g\in G$,
 the following equivalence relations hold:
\beqa
 x_0 = x_0\cdot g\quad  \Leftrightarrow\quad \ov{[x_0]}^{a} = \ov{[x_0\cdot g]}^{a}\quad \Leftrightarrow \quad \ov{ \ov{[x_0]} }^{a} = \ov{\ov{[x_0\cdot g]}}^{a}. 
\label{disjointness-relations}
\eeqa
\eep
\proof  We  use Lemma~\ref{conjugation-lemma} to prove the   relations in (\ref{disjointness-relations}).  
It is clear that $x_0 = x_0\cdot g$ implies the other two equalities. As for the opposite implications, 
let us first   suppose  that  
$\ov{[x_0]}^a = \ov{[x_0\cdot g]}^a$, i.e., $[x_0\cdot a^2]_{ (a^{-1}\cdot G_{x_0}\cdot a)   }=[x_0\cdot a\cdot g^{-1}\cdot a ]_{(a^{-1}\cdot G_{x_0}\cdot a)}$, 
for some $g\in  G$. In other words,
\beqa
x_0\cdot a^2=x_0\cdot a\cdot g^{-1}\cdot g_0\cdot a \quad \textrm{ for some }\quad g_0\in G_{x_0}.
\eeqa
Hence, $a\cdot g^{-1}\cdot g_0 \cdot a^{-1}=g_0'$ for some $g_0'\in G_{x_0}$  and therefore $g=g_0 \cdot a^{-1}\cdot (g_0')^{-1}\cdot a$. 
Since   $a^{-1}\cdot G_{x_0} \cdot a\subset G_{x_0}$, we obtain that $g\in G_{x_0}$.

Let us now suppose  that  $\ov{ \ov{[x_0]} }^{a} = \ov{\ov{[x_0\cdot g]}}^{a}$,
that is, $[x_0]_{ a^{-1}\cdot G_{x_0}\cdot a }=[x_0 \cdot g]_{ a^{-1}\cdot G_{x_0}\cdot a }$ for some $g\in  G$.   This means
\beqa
x_0=x_0\cdot g\cdot a^{-1}\cdot g_0\cdot a   \quad \textrm{ for some }\quad g_0\in G_{x_0}.
\eeqa
Thus   $g\cdot a^{-1}\cdot g_0\cdot a\in G_{x_0}$, i.e., $g=g_0'\cdot a^{-1}\cdot g_0^{-1}\cdot a$ for some $g_0'\in G_{x_0}$. 
Since   $a^{-1}\cdot G_{x_0} \cdot a\subset G_{x_0}$, we obtain again that $g\in G_{x_0}$. \qed

In the next theorem, which can be considered our main abstract result,  we identify $(x_0,a)$ from Definition~\ref{conjugation-def}
for which distinct points from  $[x_0]_G$ may give rise  to coinciding conjugate and second conjugate classes.
\bet\label{theorem}\label{main-abstract-result} Let $  S\subset G$ be a subgroup and set $R_0:=G_{x_0}\cap S$.  
Then, for all  $a\in N_{G}(R_0,S)$ and    $s\in S$,
\beqa
\, \ov{[x_0 ]}^{a} = \ov{[x_0\cdot s ]}^a\quad \textrm{and} 
\quad \ov{ \ov{[x_0 ]} }^{a} = \ov{\ov{[x_0\cdot s  ]} }^{a}. \label{first-equality}
\eeqa
Furthermore, $[x_0\cdot g]_S\subset \ov{ \ov{[ x_0\cdot g]}}^{a}$ for all $g\in G$. 
\eet
\begin{remark}\label{non-triviality-remark} Clearly, the statement of Theorem~\ref{main-abstract-result} remains true if $R_0$ is replaced
by any subgroup $R\subset R_0$ and the proof below pertains to this case. Moreover, as we will see in Subsection~\ref{curing-velocity-superselection}, the assumptions of Theorem~\ref{main-abstract-result} are compatible with $x_0\neq x_0\cdot s$ for some $s\in S$. That is, under these assumptions the first equivalence relation in (\ref{disjointness-relations}) is not true.
\end{remark}
\proof Concerning the first equality, we write  using Lemma~\ref{conjugation-lemma}
\beqa
 \ov{[x_0\cdot s]}^a=[x_0\cdot a \cdot s^{-1}\cdot a]_{a^{-1}\cdot G_{x_0}\cdot a}=[x_0\cdot r\cdot a^2]_{a^{-1}\cdot G_{x_0}\cdot a}=\ov{[x_0]}^a,
 \eeqa
where in the second step we noted that $r:=a\cdot s^{-1} a^{-1}\in R$ since $a\in N_{G}(R,S)$ and in the third step we used that $R\subset G_{x_0}$.

As for the second equality, we proceed similarly. Lemma~\ref{conjugation-lemma} gives
\beqa
\ov{\ov{[x_0\cdot s ] }}^a=[x_0\cdot s  ]_{a^{-1}\cdot G_{x_0}\cdot a}=[x_0 \cdot  a^{-1} \cdot r^{-1} \cdot a ]_{a^{-1}\cdot G_{x_0}\cdot a } =
\ov{ \ov{[x_0 ]} }^a,
\eeqa
where we made use of the fact that $r^{-1}:=a\cdot s \cdot a^{-1}\in R\subset G_{x_0}$ and  thus $r^{-1}\cdot G_{x_0}=G_{x_0}$. 

To prove the last statement, we write for any $s'\in S$
\beqa
x_0\cdot g\cdot s'=x_0\cdot g\cdot  a^{-1}\cdot (a\cdot s'\cdot a^{-1})\cdot a=x_0\cdot g\cdot a^{-1}\cdot r'\cdot a\subset [x_0\cdot g]_{a^{-1}\cdot G_{x_0}\cdot a },
\eeqa
where we used that $r'=a\cdot s' \cdot a^{-1}\in R\subset G_{x_0}$. By applying Lemma~\ref{conjugation-lemma} we conclude the proof. \qed 

\subsection{Application in representation theory of  $C^*$-algebras}\label{superselection-theory}

Let  $\mfa$ be a $C^*$-algebra and let $R\subset S\subset G$ be subgroups of the group $\Aut(\mfa)$ of automorphisms
of $\mfa$.  Furthermore, we denote by $\mrm{In}(\mfa)\subset \Aut(\mfa)$ the normal subgroup of inner automorphisms. We denote
by $P_{\mfa}\subset \mfa^*$ the set of pure states on $\mfa$ on which $\Aut(\mfa) $ acts in a natural manner. For the  resulting action of $G\subset \Aut(\mfa)$
we write
\beqa
 P_{\mfa} \times  G\ni (\om, \ga)\to \om\circ \ga\in P_{\mfa}. \label{action-on-states}
\eeqa 

In the spirit of \cite[Definition 4.1]{BR14}, we define the  set of \emph{sectors} as $X_{\mfa}:=P_{\mfa}   / \mrm{In}(\mfa)$.  We recall in Proposition~\ref{sector-proposition},
that for any  $\om_1, \om_2\in  P_{\mfa}$, the equality of sectors $[\om_1]_{\mrm{In}(\mfa)}= [\om_2]_{\mrm{In}(\mfa)}$
holds iff the GNS representations of $\om_1$, $\om_2$ are unitarily equivalent.   Furthermore, for any $\om \in P_{\mfa}$  the stabilizer group $G_{[\om]_{\mrm{In}(\mfa) }}$
is precisely the group of automorphisms from $G$ which are unitarily implementable in the GNS representation of $\om$.

Since $\mrm{In}(\mfa)\subset \Aut(\mfa)$ is a normal subgroup, (\ref{action-on-states})
gives rise to an action of $G$ on the space of sectors:
\beqa
X_{\mfa}\times G\ni ([\om]_{\mrm{In}(\mfa)}, \ga)\to [\om\circ \ga]_{\mrm{In}(\mfa)} \in X_{\mfa}.
\eeqa
Let us now fix a reference (`vacuum') state $\om_{\mrm{vac}}\in P_{\mfa}$ and a reference (`background') automorphism $\al\in G$.  
 The pair $([\om_{\vac}]_{\mrm{In}(\mfa) }, \al )\in P_{\mfa}\times G$ will play the role of  $(x_0, a)$ from Definition~\ref{conjugation-def}.  
Now Proposition~\ref{disjointness-proposition} and Theorem~\ref{main-abstract-result} give the following:
\bet\label{disjointness-proposition-alg} Let $([\om_{\vac}]_{\mrm{In}(\mfa) }, \al )$ and  $R\subset S\subset G\subset \Aut(\mfa)$ be as above
and $R\subset G_{[\om_\vac]_{\mrm{In}(\mfa) }}$.
\begin{enumerate}
\item[(a)] Suppose that $\al^{-1}\circ G_{[\om_\vac]_{\mrm{In}(\mfa) }}\circ \al\subset  G_{[\om_\vac]_{\mrm{In}(\mfa) }}$. Then, for all $\ga\in G$,  the condition  
$[\om_\vac]_{\mrm{In}(\mfa)} =[\om_\vac\circ \ga]_{\mrm{In}(\mfa)}$  is equivalent to any of the following two equalities:
\beqa 
\ov{[[\om_\vac]_{\mrm{In}(\mfa) }  ] }^{\al} = \ov{[  [\om_\vac\circ \ga]_{\mrm{In}(\mfa)}  ]   }^{\al},  \quad  \ov{ \ov{[[\om_\vac]_{\mrm{In}(\mfa)} ] }}^{\al} =\ov{ \ov{[  [\om_\vac\circ \ga]_{\mrm{In}(\mfa)}  ]   }}^{\al}.
\eeqa

\item[(b)] Suppose that $\al\in N_{G}(R,S)$.  Then for all $\ga\in S$ 
 \beqa 
\ov{[[\om_\vac]_{\mrm{In}(\mfa) }  ] }^{\al} = \ov{[  [\om_\vac\circ \ga]_{\mrm{In}(\mfa)}  ]   }^{\al} \textrm{ and }  \ov{ \ov{[[\om_\vac]_{\mrm{In}(\mfa)} ] }}^{\al} =
\ov{ \ov{[  [\om_\vac\circ \ga]_{\mrm{In}(\mfa)}  ]   }}^{\al} \textrm{ holds. }
\eeqa

\item [(c)] Suppose that $\al\in N_{G}(R,S)$. Then, for all $\ga\in G$,  $[\om_\vac\circ \ga]_{\mrm{In}(\mfa) \cdot S } \subset \ov{ \ov{[[\om_\vac\circ \ga]_{\mrm{In}(\mfa)} ] }}^{\al}$.
\end{enumerate}
\eet
\nin As indicated in Remark~\ref{non-triviality-remark}, part (b) of this theorem is only non-trivial if $S\backslash G_{[\om_\vac]_{\mrm{In}(\mfa)} } \neq \emptyset$. 
This latter relation fails, in particular, under the combined assumptions of (a) and (b), which give $S\subset \al^{-1}\circ R\circ \al\subset  \al^{-1}\circ G_{ [\om_\vac]_{\mrm{In}(\mfa)}  }\circ \al\subset G_{  [\om_\vac]_{\mrm{In}(\mfa)}  }$.
We will  give  examples illustrating the non-trivial case in  Subsection~\ref{curing-velocity-superselection}.
\section{Relative normalizers in the theory of canonical commutation relations} \label{CCR-theory}
\setcounter{equation}{0}
\subsection{Relative normalizers in the inhomogeneous symplectic group} \label{semi-direct}

Given an infinite dimensional real vector space $\mcL$ we denote its algebraic dual by  $\mcL^*$ and the action of an element $\vv\in \mcL^*$ on $\vf\in \mcL$
by $(\vv, \vf)$. For any {\bc real} linear map $T: \mcL\to \mcL$ its transposition $T^t: \mcL^*\to \mcL^*$ is defined by  $(T^t\vv, \vf)=(\vv, T\vf)$. 
{\bc The group of invertible liner maps on $\mcL$ is denoted $\mrm{GL}(\mcL)$}.
We equip the vector space $\mcL$ with a non-degenerate symplectic form $\sym(\,\cdot, \cdot \,)$ and say that $T\in \mrm{GL}(\mcL)$ is symplectic if 
$ \sym( \,T\vf_1 , \, T\vf_2 ) = \sym( \vf_1 , \, \vf_2)$  for all $\vf_1, \vf_2\in \mcL$. The group of symplectic maps on $\mcL$ 
is denoted $\mrm{Sp}(\mcL)$. We define the inhomogeneous symplectic group in a way which is suitable for our purposes{\bc, namely}
\beqa
\mrm{ISp}(\mcL):= \mcL^* \rtimes_{\vp} \mrm{Sp}(\mcL)
\eeqa
with  the group homomorphism $\vp: \mrm{Sp}(\mcL)\to \mrm{GL}(\mcL^*)$ given by $\vp(T)=(T^{-1})^t$. The elements of $\mrm{ISp}(\mcL)$ are pairs
 $g=(\vv, T)\in  \mcL^* \times  \mrm{GL}(\mcL)$ and the product is defined by   $(\vv_1, T_1)\cdot (\vv_2, T_2)=(\vv_1+ (T_1^{-1})^t\vv_2, T_1T_2)$.
We write $T:=(0,T)$, $\vv:=(\vv, I)$ and  treat $\mcL^*$ and $\mrm{Sp}(\mcL)$ as subgroups of $\mrm{ISp}(\mcL)$. In this spirit we write
\beqa
\mcL^*_\mrm{R}\subset \mcL^*_\mrm{S}\subset    \mrm{ISp}(\mcL), \label{R-S-subspaces}
\eeqa
for subspaces $ \mcL^*_\mrm{R}, \mcL^*_\mrm{S}\subset \mcL^*$ which we treat as abelian subgroups of $\mrm{ISp}(\mcL)$.  
Since the group relations give $T\cdot \vv\cdot T^{-1}=(T^{-1})^t\vv$, for all $T\in  \mrm{Sp}(\mcL)$, $\vv\in \mcL^*$, we immediately obtain:
\bel\label{symplectic-lemma}  $T\in N_{ \mrm{ISp}(\mcL) }(\mcL^*_\mrm{R} , \mcL^*_\mrm{S})$ iff $  (T^{-1})^t\, \mcL^*_\mrm{S} \subset   \mcL^*_\mrm{R}$.
\eel
\nin For the above considerations $\mrm{Sp}(\mcL)$ could be replaced with any other subgroup of $\mrm{GL}(\mcL)$, but
the symplectic structure will be important in the next subsection.

\subsection{Relative normalizers in automorphism groups of CCR algebras}\label{algebraic-preliminaries}

Let us first summarize some relevant information from the theory of  CCR algebras and their Bogoliubov automorphisms referring to \cite{DG,Ru78, Ro70} for more complete treatment.
The  $C^*$-algebra $\mfa$ of canonical commutation relations, associated with $\mcL$, is constructed in a standard manner: The $*$-algebra generated by symbols $\{W(\vf)\}_{\vf\in \mcL }$ satisfying the Weyl relations
\beqa
W(\vf_1)W(\vf_2)=e^{-i\sym( \vf_1, \vf_2)} W(\vf_1+\vf_2), \quad W(\vf)^*=W(-\vf) \label{Weyl-relations}
\eeqa
is completed in the $C^*$-norm  $\|\,\cdot\,\|=\sup_{(\pi, \hil)}\| \pi(\,\cdot\,)\|_{B(\hil)}$, where the supremum is taken over all representations 
$(\pi, \hil)$.

Proceeding to  {\bc relevant} automorphisms of $\mfa$, we consider a group homomorphism $\al:~\mrm{ISp}(\mcL)\to \Aut(\mfa)$ defined on the Weyl operators by
\beqa
\al_{(\vv, T)}(W(\vf))=e^{-2i(\vv, T\vf)}W(T\vf){\bc.} \label{automorphism-definition}
\eeqa
{\bc It is} extended to $\mfa$ using  the uniqueness theorem of Slawny, see \cite[Theorem~2.1 and page~13]{Petz}, and the boundedness of automorphisms
 of $C^*$-algebras \cite[Proposition~2.3.1]{BR}.

Lemmas~\ref{homomorphism-lemma} and \ref{symplectic-lemma}  give the following criterion for the existence of non-empty relative normalizers for subgroups of $ \Aut(\mfa)$.
\bel\label{Aut-rel-norm} Consider the abelian subgroups $ \mcL^*_\mrm{R}\subset \mcL^*_\mrm{S}\subset    \mrm{ISp}(\mcL)$ as in (\ref{R-S-subspaces}).
Let  $R:=\al_{\mcL^*_\mrm{R}}$,  $S:=\al_{\mcL^*_\mrm{S}}$ and $G\subset \Aut(\mfa)$ be any subgroup containing $\al_{\mrm{ISp}(\mcL)}$.
Then for any $T\in \mrm{ISp}(\mcL)$ the following implication holds:
\beqa
(T^{-1})^t\, \mcL^*_\mrm{S} \subset   \mcL^*_\mrm{R}\quad \Rightarrow \quad \al_{T} \in N_{G}(R,S).
\eeqa
\eel

To be able to apply Theorem~\ref{disjointness-proposition-alg} we  need to choose the `vacuum' state $\om_{\text{vac}}$. 
We do it in the standard manner: Suppose that $\mcL$ is a dense subspace of a Hilbert space $\mfh$ with a scalar 
product $\lan\,\cdot \,,\, \cdot \,\ran$ such that  $\mathbf{\si}(\, \cdot\,, \, \cdot\,)=\mrm{Im}\lan\,\cdot\, ,\, \cdot \, \ran$
and define on the Weyl operators
\beqa
\om_{\text{vac}}(W(\vf)):=e^{-\h \|\vf\|^2}, \quad \vf\in \mcL. \label{vacuum-state-def}
\eeqa 
The resulting GNS representation $\pi_{{\text{vac}}}$, which can be chosen to act on the Fock space $\mcF=\Ga(\mfh)$ (see (\ref{Fock-space}) below),  is faithful and irreducible. 
We have
\beqa
W_{\text{vac}}(\vf):=\pi_{{\text{vac}}}(W(\vf))= e^{a^*(\vf)-a(\vf)}, \label{Weyl-operator}
\eeqa
where $a^*(\,\cdot\,), a(\, \cdot\,)$ are the creation and annihilation operators on $\mcF$. A concrete choice of $\mfh$ will be made in the next section.

\section{Infravacua in  QED as elements of relative normalizers} \label{QED-section}

\subsection{The model}\label{sec:model} 
 \setcounter{equation}{0}

We consider  one spinless non-relativistic electron interacting with the second-quantized electromagnetic field in the setting of non-relativistic quantum electrodynamics (QED). For a textbook presentation see \cite{Sp}, we follow here mostly \cite{CFP09}. 

We set $L^2(\real^3;\complex^3):=L^2(\real^3)\otimes \complex^3$   and denote the scalar product by $\lan \,\cdot\,, \, \cdot \, \ran$. The single-photon Hilbert space $\mfh$ is the following space of transverse functions
\beqa
L^2_{\tr}(\real^3;\complex^3):=\{ \vf\in L^2(\real^3;\complex^3)\,|\,  \pmb{k} \cdot \pmb{f}(\pmb{k}) =0 \;\; \text{a.e.}\}
\eeqa 
and we denote by $P_{\tr}: L^2(\real^3;\complex^3)\to L^2(\real^3;\complex^3)$  the orthogonal projection on $L^2_{\tr}(\real^3;\complex^3)$.
We {\bc write} $S^2$  for the unit sphere in $\real^3$
and introduce the polarisation vectors $\real^3\ni \vk\mapsto  \pmb{\epsilon}_{\pm}(\vk)\in S^2$, given by, e.g., \cite{LL04}
\begin{eqnarray}
\pmb{\epsilon}_+ ( \pmb{k})= \frac{(k_2, -k_1,0)}{\sqrt{k_1^2 + k_2^2}}, \quad
\pmb{\epsilon}_-( \pmb{k})= \frac{\pmb{k}}{|\pmb{k}|} \times \pmb{\epsilon}_+(\pmb{k}),
\end{eqnarray}
which satisfy $\pmb{k} \cdot \pmb{\epsilon}_\pm(\vk) =0$ and $\pmb{\epsilon}_+(\vk) \cdot \pmb{\epsilon}_-(\vk) =0$ for $\vk=(k_1,k_2,k_3)\in\real^3$. In  terms of these vectors we can write, in the $L^2$-sense,
\beqa
(P_{\tr}\vf)(\vk)=\sum_{\la=\pm }  \big(\vf(\vk)\cdot  \pmb{\epsilon}_\la (\pmb{k} ) \big) \pmb{\epsilon}_\la ( \pmb{k}). \label{transverse-projection}
\eeqa
Next, we denote
by $\mcF$ the symmetric Fock space over $ \mfh=L^2_{\tr}(\real^3; \complex^3)$, which is the fiber Fock space pertaining to the decomposition~(\ref{fiber-hamiltonians}).  More precisely,
\begin{equation}\label{Fock-space}
{\mathcal{F}} := \oplus_{n=0}^\infty {\mathcal{F}}^{(n)}, \quad {\mathcal{F}}^{(n)} :=  \operatorname{Sym}_n (  \mfh  \phantom{}^{\otimes n}),\quad 
{\mathcal{F}}^{(0)} = \mathbb{C}\Omega.
\end{equation}
We define  the quantized electromagnetic vector potential\footnote{up to a normalization constant, which is absorbed
into $\alt^{1/2}$ in (\ref{Hamiltonian-section1}). } with infrared and ultraviolet cut-offs $0\leq \si\leq \ka$  
as an operator on a certain domain in $\mcF$
\beqa
\vA_{[\si,\ka]}(\vx):=\sum_{\la=\pm} \int \fr{d^3k}{\sqrt{{\bc |\vk|}}} \chi_{[\si, \ka]}(|\vk|) \pol(\vk)\big( e^{-i\vk\cdot \vx} a^*_{\la}(\vk) +   e^{i\vk\cdot \vx} a_{\la}(\vk)  \big).
\eeqa
Here $\chi_{\De}$ denotes the characteristic function of a set $\De$ and $a_{\la}(\vk), a^{*}_{\la}(\vk)$ are the standard (improper) creation and annihilation
operators on $\mcF$ such that $[a_{\la}(\vk), a^*_{\la'}(\vk')] =\de_{\la\la'}\de(\vk-\vk')$  and all other commutators vanish.  They are related
to the creation and annihilation operators appearing in (\ref{Weyl-operator}) via 
$a^*(\vf)=\sum_{\la=\pm} \int d^3k\, a_{\la}^*(\vk)\, (\pol(\vk)\cdot \vf(\vk))$,  for $\vf\in \mfh$.

Furthermore, we define the free 
Hamiltonian and momentum operators of the electromagnetic field
\beqa
H_{\pho}=\sum_{\la=\pm} \int d^3k\, |\vk| \, a^*_{\la}(\vk) a_{\la}(\vk), \quad \vP_{\pho}=\sum_{\la=\pm} \int d^3k\, \vk \, a^*_{\la}(\vk) a_{\la}(\vk).
\eeqa
The fiber Hamiltonians, which appeared in the decomposition (\ref{fiber-hamiltonians}), are given by  
\beqa
H_{\vP,\si}=\h(\vP-\vP_{\pho}+ \alt^{1/2}\vA_{[\si,\ka]}(0))^2+H_{\pho},\quad H_{\vP}:=H_{\vP,\si=0}. \label{Hamiltonian-section1}
\eeqa
They are self-adjoint, positive operators on a common domain independent of $\vP$ (see, e.g., \cite{Sp, Hi00, KM14}). 
We denote by $E_{\vP,\si}:=\mrm{inf}\,\mrm{ Spec}(H_{\vP,\si})$, $E_{\vP}:=\mrm{inf}\,\mrm{ Spec}(H_{\vP})$ the respective 
infima of the spectra of $H_{\vP,\si}$, $H_{\vP}$.  They are rotation invariant functions of $\vP$.

Now we recall some spectral results, mostly from \cite{CFP09, FP10},   which we will use  below. From now on we discuss
the regime of low coupling $\ti\al$ and we are interested in 
momenta $\vP$ restricted to the ball
\beqa
\mcS=\{ \vP\in \real^3\,|\, |\vP|< \tfrac{1}{3} \}.
\eeqa 
It is well  known that for any $\si>0$ the operators $H_{\vP,\si}$ have ground-states $\Psi_{\vP,\si}\in \mcF$, $\|\Psi_{\vP,\si}\|=1$, corresponding to
isolated eigenvalues $E_{\vP,\si}$. The dependence $\vP\mapsto E_{\vP,\si}$ is analytic for any fixed $\si>0$
by Kato perturbation theory.   In the limit $\si\to 0$, as the spectral gap closes,  $\Psi_{\vP,\si}$  tend weakly to zero 
\cite{CFP09,  Fr73, Fr74.1, Ch00} and the Hamiltonians $H_{\vP}$ do not have ground-states for $\vP\neq 0$ \cite{HH08}.  
To analyze  this phenomenon, one introduces auxiliary vectors
\beqa
\Phi_{\vP,\si}:=W_\vac(\vv_{\vP, \si})\Psi_{\vP,\si}, \quad W_\vac(\vv_{\vP,\si})=e^{a^*( \vv_{\vP,\si})-a(\vv_{\vP,\si})   }, \label{modified-states}
\eeqa
where $\vv_{\vP,\si}$ has the form
\beqa
\pmb{v}_{\pmb{P},\sigma}(\pmb{k}) =   \alt^{1/2}  P_{\tr} \frac{\chi_{[\sigma, \kappa]} (|\pmb{k}|) }{|\pmb{k}|^{3/2}}
\frac{\nabla E_{\pmb{P},\sigma}}{1  - \hat{\pmb{k}}\cdot \nabla E_{\pmb{P},\sigma}}, \label{v-P-sigma}
\eeqa
{\bcc and} we set $\hat{\vk}:=\vk/|\vk|$ and $\nabla E_{\pmb{P},\sigma}:=\nabla_{\vP}E_{\pmb{P},\sigma}$. The following lemma collects some facts from \cite{CFP09, FP10}.\footnote{Precisely, for (a) and (b) see \cite[Theorem III.3 and Corollary III.4]{FP10}, for (c) see \cite[Eq.~(III.2) and  formula (V.6)]{CFP09} and for (d) \cite[Theorem III.1]{CFP09}.}
\bel\label{spectral} Let $\ti\al>0$ be sufficiently small and $\vP\in \mcS$. Then
\begin{enumerate}
\item[(a)] The function $\vP \mapsto E_{\vP}$ is rotation invariant, twice differentiable and has  strictly positive second derivative with respect to $|\vP|$. 
\item[(b)] $\lim_{\si\to 0} \pa_{\vP}^{\be} E_{\vP,\si}$ exists and equals $\pa_{\vP}^{\be} E_{\vP}$ for $|\be|\leq 2$.
\item[(c)] $|\nabla E_{\vP,\si}|\leq v_{\mathrm{max}}<1$ and $|\nabla E_{\vP}|\leq v_{\mathrm{max}}<1$ for some constant $v_{\mathrm{max}}$, uniformly in $\sigma$ and in $\vP \in \mcS$.
\item[(d)] $\Phi_{\vP}:=\lim_{\si\to 0}\Phi_{\vP,\si}$ exists in norm for a suitable choice of the phases of $\Psi_{\vP,\si}$. 
\end{enumerate}
\eel
\nin In the following we assume that the phases of $\Psi_{\vP,\si}$ are fixed as in Lemma~\ref{spectral}~(d). Using Lemma~\ref{spectral}~(b) 
we can define the pointwise limit
\beqa
\vv_{\vP}(\vk):=\lim_{\si\to 0}\vv_{\vP,\si}(\vk)=\frac{\alt^{1/2}}{ |\pmb{k}|^{3/2}}  P_{\tr} \chi_{[0, \kappa]} (|\pmb{k}|) 
\frac{ \nabla E_{\pmb{P}} }{1  - \hat{\pmb{k}}\cdot \nabla E_{\pmb{P}}}. \label{v-P}
\eeqa
We note that the expressions $1-\hat{\pmb{k}}\cdot \nabla E_{\pmb{P}, \si}$ and  $ 1-\hat{\pmb{k}}\cdot \nabla E_{\pmb{P}}$ 
in the denominators of (\ref{v-P-sigma}) and (\ref{v-P}) are different from zero by Lemma~\ref{spectral}~(c). 
The fact that   $\vv_{\vP}$ is not in $L^2_{\tr}(\real^3; \complex^3)$ for $0\neq \vP\in \mcS$ will be important below. {\bc We will also use that {\bcc $\vv_{\vP=0}=0$, which is a consequence of rotational invariance.}

\subsection{Infravacua as elements of relative normalizers}\label{relative-norm-subsection}

In this subsection we will give a concrete realization of the structure $\mcL^*_\mrm{R}\subset \mcL^*_\mrm{S}\subset    \mrm{ISp}(\mcL)$
which appeared in  (\ref{R-S-subspaces}). 
The symplectic space $\mcL$, which we will use in the following analysis, is defined as follows: 
\beqa\label{mcL}
\mcL:=\bigcup_{\thetae>0} L^2_{\tr, \thetae}(\real^3; \complex^3),     
\eeqa
where $L^2_{\tr, \thetae}(\real^3; \complex^3):=\{\,\vf\in L^2_{\tr}(\real^3; \complex^3)\,|\, \vf(\vk)=0 \textrm{ for } |\vk|\leq \thetae\,\}$.  
The symplectic form is given by $\pmb{\si}(\,\cdot\,, \, \cdot\,  )=\mrm{Im}\lan \,\cdot\,, \, \cdot \, \ran$.
We introduce the following subspaces of $\mcL^*$, 
\beqa
& &\mcL^*_{\mrm{S}}:=\mcL^*_{\mrm{D}}+\mcL^*_{\mrm{R}},   \textrm{ where } 
\mcL^*_{\mrm{D}}:=\Span_{\real}\{\, \vv_{\vP}  \,|\, \vP\in \mcS\,\}, \ \  \mcL_{\mrm{R}}^*:= L^2_{\tr}(\real^3; \complex^3)_{\real},\quad\quad \label{cloud-definitions}
\eeqa
and the \emph{dressing functions} $\vv_{\vP}$ appeared in (\ref{v-P}). Here $L^2_{\tr}(\real^3; \complex^3)_{\real}$ denotes the subspace of real-valued functions  in $L^2_{\tr}(\real^3; \complex^3)$
and the linear spans in  (\ref{cloud-definitions}) are over the field of real numbers. 
Furthermore, we set here $(\vv, \vf):=\mrm{Im}\lan \vv, \vf\ran$ for  $\vv\in \mcL_{\mrm{S}}^*$,  $\vf\in \mcL$, which
 is well-defined since all $\vf\in \mcL$ vanish near zero.

In the following we  exhibit maps $T\in \mrm{Sp}(\mcL)$ such that $(T^{-1})^t\, \mcL^*_\mrm{S} \subset   \mcL^*_\mrm{R}$.
 By Lemmas~\ref{symplectic-lemma}, \ref{Aut-rel-norm}  such maps, {\bc and the corresponding automorphisms $\al_{T}$}, are elements of  relative normalizers. 
More precisely, we will show in Theorem~\ref{infravacuum-theorem}, that \emph{infravacuum maps}  first introduced by  Kraus, Polley and Reents \cite{Re74,KPR77, Kr82}
have the above mapping property. For this purpose, in Proposition~\ref{theo:convergenceseq} we collect the 
essential features of the infravacuum maps which can be 
found in the literature, up to some technical mismatches
relating, e.g., to different choices of the symplectic space  and different implementation of the infrared regularization. 
We postpone the somewhat lengthy definition
of the Kraus-Polley-Reents maps and the proof of Proposition~\ref{theo:convergenceseq} to Section~\ref{Map-T}.
At the same time our proof of part (a) of this proposition is new and  easier to follow than the corresponding arguments 
available in \cite{Re74,KPR77, Kr82, Ku98}. 
\bep\label{theo:convergenceseq} There exist $T\in \mrm{Sp}(\mcL)$ such that{\bc:}
\begin{enumerate}
\item[(a)] The limit $T\vv_{\vP}:=\lim_{n\to \infty} T \vv_{\vP,\si_n}$ exists in $L^2_{\mrm{tr}}(\real^3;\complex^3)$
for a certain subsequence $\{\si_n\}_{n\in\nat}$ tending to zero.
\item[(b)] $\|T\vf\|\leq c\|\vf\|$ for all \emph{real-valued} $\vf\in\mcL$ and {\bc some} $c$ independent of $\vf$.
\end{enumerate}
Such $T$ are called \emph{infravacuum maps}.
\eep
As noted below formula~(\ref{v-P}), the functions $\vv_{\vP,\si}$  escape from $L^2_{\tr}(\real^3;\complex^3)$ in the limit $\si\to 0$.  
Therefore {\bc $T\vv_{\vP}$ above should be understood as one symbol and} part (a) of Proposition~\ref{theo:convergenceseq} does not follow from part (b).   {\bc Instead, it} demonstrates a remarkable regularizing property
of $T$. 
We stress that this feature is restricted to real-valued functions and an infravacuum map cannot be complex-linear.
This is a consequence of the symplectic property and the following computation
\beqa
\mrm{Im}\lan T\vv_{\vP,\si_n}, T (i\vv_{\vP,\si_n})\ran=\mrm{Im}\lan \vv_{\vP,\si_n}, i\vv_{\vP,\si_n}\ran=\|  \vv_{\vP,\si_n}\|^2 \underset{n\to\infty}{\to}   \infty,
\eeqa
which shows that $T (i\vv_{\vP,\si_n})$ diverges  in $L^2_{\mrm{tr}}(\real^3;\complex^3)$ as $n\to \infty$.  
Another  consequence of  properties (a), (b) from Proposition~\ref{theo:convergenceseq} is that {\bc the} infravacuum maps are elements of relative normalizers, 
which is the main result of this section.
\bet\label{infravacuum-theorem}  The infravacuum maps $T\in \mrm{Sp}(\mcL)$ from Proposition~\ref{theo:convergenceseq}  satisfy 
\beqa
(T^{-1})^t\mcL^*_{\mrm{S}}\subset  \mcL^*_{\mrm{R}}, \label{infravacuum-mapping-property}
\eeqa
where the  subspaces   $\mcL^*_{\mrm{R}}\subset \mcL^*_{\mrm{S}}$ are given by (\ref{cloud-definitions}). Hence $T\in  N_{ \mrm{ISp}(\mcL) }(\mcL^*_\mrm{R} , \mcL^*_\mrm{S})$
and $\al_{T} \in N_{G}(R,S)$. Here $R:=\al_{\mcL^*_\mrm{R}}$,  $S:=\al_{\mcL^*_\mrm{S}}$ and $G\subset \Aut(\mfa) $ is any
subgroup containing  $\al_{\mrm{ISp}(\mcL)}$. The homomorphism $\al$ is
given by (\ref{automorphism-definition}).
\eet
\proof A general element $\vv_{\mrm{S}}\in \mcL^*_{\mrm{S}}$ has the form $\vv_{\mrm{S}}=\vv_{\mrm{D}}+\vv_{\mrm{R}}$, where 
$\vv_{\mrm{D}}\in \mcL^*_{\mrm{D}}$ and $\vv_{\mrm{R}}\in \mcL^*_{\mrm{R}}:=L^2_{\mrm{tr}}(\real^3;\complex^3)_{\real}$.
Thus we have for any $\vf\in \mcL$
\beqa
((T^{-1})^t\vv_{\mrm{S}}, \vf)=(\vv_{\mrm{S}}, T^{-1}\vf )=\mrm{Im}\lan \vv_{\mrm{D}}, T^{-1}\vf\ran+ \mrm{Im}\lan \vv_{\mrm{R}}, T^{-1}\vf\ran.
\label{T-verification}
\eeqa

Clearly, we have $\vv_{\mrm{D}}=\sum_{i=1}^Nc_i\vv_{\vP_i}$ for some $c_i\in \real$, $N\in \nat$, $\vP_i\in \mcS$.  We define 
accordingly its approximant $\vv_{\mrm{D}, \si}:=\sum_{i=1}^Nc_i\vv_{\vP_i, \si}\in \mcL$ and write, using the sequence $\{\si_n\}_{n\in \nat}$ 
from Proposition~\ref{theo:convergenceseq},
\beqa
\mrm{Im}\lan \vv_{\mrm{D}}, T^{-1}\vf\ran=\lim_{n\to \infty} \mrm{Im}\lan \vv_{\mrm{D}, \si_n}, T^{-1}\vf\ran=\lim_{n\to \infty} \mrm{Im}\lan T\vv_{\mrm{D}, \si_n}, \vf\ran=( T\vv_{\mrm{D}}, \vf). \label{symplectic-computation-one}
\eeqa
Here in the first step we used Lemma~\ref{spectral} (b) and the fact that $T^{-1}\vf$ vanishes in some neighbourhood of zero. In the second
step we used that $T$ is symplectic and in the last step we applied Proposition~\ref{theo:convergenceseq} (a). This latter statement also tells
us that $ T\vv_{\mrm{D}}:=\lim_{n\to\infty} T\vv_{\mrm{D}, \si_n}\in \mcL^*_{\mrm{R}}$.

The second term on the right hand side of (\ref{T-verification}) is handled by a similar and simpler consideration: We define 
$\vv_{\mrm{R}}^{\si}(\vk):=\chi_{[\si, \infty)}(|\vk|) \vv_{\mrm{R}}(\vk)$ so that  $\vv_{\mrm{R}}=\lim_{\si\to 0} \vv_{\mrm{R}}^{\si}$
in the norm topology of $ L^2_{\mrm{tr}}(\real^3;\complex^3)$. Since the functions $\vv_{\mrm{R}}^{\si}$
are real-valued, Proposition~\ref{theo:convergenceseq}~(b) gives the existence of $T\vv_{\mrm{R}}:=\lim_{\si\to 0} T \vv_{\mrm{R}}^{\si}\in \mcL^*_{\mrm{R}}$.
Now the proof of (\ref{infravacuum-mapping-property}) is completed by a computation analogous to (\ref{symplectic-computation-one}). 

The last statement of the theorem is a consequence of {\bc Lemmas~\ref{symplectic-lemma}, \ref{Aut-rel-norm}}. \qed

\subsection{Infravacua and velocity superselection} \label{curing-velocity-superselection}

\textcolor{black}{Proceeding towards the problem of velocity superselection,  we define the following states on the Weyl algebra $\mfa$ over the symplectic space $\mcL$ as introduced in \eqref{mcL}. }
\beqa
\om_{\pmb{P}}(A):=\lim_{\sigma\to 0}\lan \Psi_{\pmb{P},\sigma}, \pi_{{\text{vac}}}(A)  \Psi_{\pmb{P}, \sigma}\ran=\lan \Phi_{\pmb{P}}, \pi_{{\text{vac}}}(\al_{\vv_{\pmb{P} }  }(A))  \Phi_{\pmb{P}}\ran,
\quad A\in\mfa. \label{state-formula}
\eeqa
\textcolor{black}{These states describe plane-wave configurations of the electron with velocity~$\nabla E_{\vP}$.}
Here $\pi_{{\text{vac}}}$ is the vacuum representation defined in (\ref{vacuum-state-def}) and in the second step we used  (\ref{modified-states}), the Weyl relations (\ref{Weyl-relations}),  Lemma~\ref{spectral}~(d), definition~(\ref{automorphism-definition}) and the specifications
$(\vv, \vf)=\mrm{Im}\lan \vv, \vf \ran$, $\pmb{\si}(\vf_1, \vf_2)=\mrm{Im}\lan \vf_1, \vf_2\ran$  for $\vv\in \mcL^*_{\mrm{S}}$, $\vf, \vf_1,\vf_2\in \mcL$,  which appeared below (\ref{cloud-definitions}) and above (\ref{vacuum-state-def}), respectively. As  $\om_{\pmb{P}}$  are pure states on $\mfa$ (cf. \cite[Corollary 10.2.5]{KR}), we can use the framework
of Subsection~\ref{superselection-theory} to study the corresponding superselection structure. The starting point is the following proposition, whose proof
is given in Section~\ref{proof-of-velocity-superselection}.  
\bep\label{velocity-superselection-x} We have  $[\om_{\vP}]_{\sect}  \neq [\om_{\vP'}]_{\sect}$ for all $\vP, \vP'\in \mcS$ such that $\vP\neq \vP'$.

\eep

We recall that  for the present model, and a similar model describing the electron with spin,  disjointness of $ [\om_{\vP=0}]_{\sect}  $ and $[\om_{\vP'}]_{\sect}$, $\vP'\neq 0$ was
shown in \cite{CFP09, CF07} by exploiting the absence of the number operator in non-Fock representations. In the Nelson model and in 
a semi-relativistic model of QED 
disjointness for all $\vP\neq \vP'$ (from suitable balls around zero) was verified in \cite{Fr73,KM14} using theory of infinite tensor products of Hilbert spaces. Our proof in Section~\ref{proof-of-velocity-superselection},
inspired by \cite[Lemma 2.2]{Ku98}, exhibits central sequences in $\mfa$ which can distinguish $[\om_{\vP}]_{\sect}$ from $[\om_{\vP'}]_{\sect}$. We think this argument is quite simple and intuitive. 

Our main result, concerning the problem of velocity superselection in the model of non-relativistic QED, is stated below.
The {\bc conjugate classes are}  defined using an arbitrary subgroup  $ \al_{\mrm{ISp}(\mcL)} \subset G \subset \Aut(\mfa) $,  
the Fock vacuum  $[\om_{\text{vac}}]_{\sect}$ as the reference sector  and two different types of background automorphisms:
In part (a) we consider `regular' backgrounds in which case  the conjugation procedure does not merge the disjoint sectors from Proposition~\ref{velocity-superselection-x} into coinciding classes. In part~(b) we show that such {\bc a} merging effect  is achieved{\bc,} if the infravacuum automorphisms are used as a background.     
\bet\label{disjointness-proposition-nr-qed}  For the family of states $\{\om_{\vP}\}_{\vP\in \mcS}$ defined in (\ref{state-formula}) the following is true:
\begin{enumerate}
\item[(a)] Suppose that $\be\in G\subset \mrm{Aut}(\mfa)$ is unitarily implemented in  $\pi_{{\mrm{vac}}}$ (e.g., $\be=\mrm{id}$). Then, 
for all $0\neq \vP'\in \mcS$, 
\beqa 
\ov{[[\om_{\vP=0}]_{\mrm{In}(\mfa) }  ] }^{\be} \neq \ov{[  [\om_{\vP'}]_{\mrm{In}(\mfa)}  ]   }^{\be} \textrm{ and }  \ov{ \ov{[[\om_{\vP=0}]_{\mrm{In}(\mfa)} ] }}^{\be} \neq\ov{ \ov{[  [\om_{\vP'}]_{\mrm{In}(\mfa)}  ]   }}^{\be} \textrm{ holds. }
\eeqa

\item[(b)] Let $T\in \mrm{Sp}(\mcL)$ be an infravacuum map (cf. Proposition~\ref{theo:convergenceseq}) and $\al_{T}\in \Aut(\mfa)$ be
given by (\ref{automorphism-definition}). Then, for all $\vP, \vP'\in \mcS$,
 \beqa 
\ov{[[\om_{\vP}]_{\mrm{In}(\mfa) }  ] }^{\al_T} = \ov{[  [\om_{\vP'}]_{\mrm{In}(\mfa)}  ]   }^{\al_T} \textrm{ and }  \ov{ \ov{[[\om_{\vP}]_{\mrm{In}(\mfa)} ] }}^{\al_T} =
\ov{ \ov{[  [\om_{\vP'}]_{\mrm{In}(\mfa)}  ]   }}^{\al_T} \textrm{ holds. }
\eeqa

\end{enumerate}
\eet
\proof We consider the subspaces $\mcL^*_{\mrm{R}}$, $\mcL^*_{\mrm{S}}$ defined in (\ref{cloud-definitions}){\bc,}  set 
$R:=\al_{\mcL^*_\mrm{R}}$,  $S:=\al_{\mcL^*_\mrm{S}}$ and fix  $ \al_{\mrm{ISp}(\mcL)} \subset G \subset \Aut(\mfa) $ as in Theorem~\ref{infravacuum-theorem}. 
Now parts (a) and (b) can be inferred from the corresponding parts of Theorem~\ref{disjointness-proposition-alg} as follows:

(a) We recall that $G_{[\om_{\text{vac}}]_{\sect}}$ coincides with the group of automorphisms from $G$ which are unitarily implemented in the vacuum representation (cf. Theorem~\ref{sector-proposition}), thus $\be^{-1}\circ G_{[\om_{\text{vac}}]_{\sect}}\circ \be=G_{[\om_{\text{vac}}]_{\sect}}$. Furthermore,  we obtain from  formula~(\ref{state-formula}) 
\beqa\label{inner}
\om_{\vP}(A)=\lan \Om, \pi_{{\text{vac}}}(U_{\vP} \al_{\vv_{\vP}}(A) U_{\vP}^*)\Om\ran=\lan \Om, \pi_{{\text{vac}}}( \al_{\vv_{\vP}}(\ti{U}_{\vP}A\ti{U}^*_{\vP}) )\Om\ran, \quad A\in \mfa.
\eeqa
 Here we found a unitary $U_{\vP}\in \mfa $ such that $\Phi_{\vP}=\pi_{{\text{vac}}}(U_{\vP}^*)\Om$  (cf.  \cite[Theorem 10.2.1]{KR})
 and set $\ti{U}_{\vP}:=\al_{\vv_{\vP}}^{-1}(U_{\vP})$. Therefore, 
 \beqa
 \, [\om_{\vP}]_{\sect}=[\om_{\text{vac}}\circ \al_{\vv_{\vP}}]_{\sect} \quad \textrm{and} \quad [\om_{\vP=0}]_{\sect}=[\om_{\text{vac}}]_{\sect}, \label{equality-of-sectors}
 \eeqa
where the second equality follows from  $\vv_{\vP=0}=0$ (cf. definition~(\ref{v-P})). Given these identifications, 
 the statement follows from Proposition~\ref{velocity-superselection-x} and Theorem~\ref{disjointness-proposition-alg} (a).
 
 (b) By Theorem~\ref{infravacuum-theorem}, we have $\al_{T}\in N_G(R,S)$. Now, since $\al_{\vv_{\vP}}\in S$, the first equality in (\ref{equality-of-sectors}) 
and Theorem~\ref{disjointness-proposition-alg} (b) give the claim. The statement concerning the second conjugate classes can also be obtained
from Theorem~\ref{disjointness-proposition-alg} (c) noting that $[\om_{\vP}]_{\sect\cdot S}=[\om_{\text{vac}}]_{\sect \cdot S}$ for any $\vP\in \mcS$. \qed

 Theorem~\ref{disjointness-proposition-nr-qed} (b) shows that the effect of velocity superselection is eliminated at the level of 
  the   conjugate
classes with respect to  the infravacuum background. It turns out that the infravacuum automorphisms can also be used to cure velocity superselection at the level of
states:
By formula~(\ref{state-formula}) we can write $\om_{\pmb{P}} =  \om_{\Phi_{\vP}} \circ  \al_{\vv_{\pmb{P} }  }$, where we set $\om_{\Phi}(\,\cdot\,) :=\lan \Phi, \pi_{{\text{vac}}}(\,\cdot \,) \Phi\ran$ for any unit vector $\Phi\in \mcF$.
We define
\beqa
\om_{\vP,T}:=\om_{\Phi_{\vP}} \circ\al_{T}\circ \al_{\vv_{\vP}},  \label{state-one}
\eeqa
where $T$ is an infravacuum map and $\vP\in \mcS$. {\bc These} are modifications of the states $\om_{\pmb{P}}$ above by inserting the infravacuum between the  state $\om_{\Phi_{\vP}}$ of the `undressed electron' and the automorphism $ \al_{\vv_{\pmb{P} }  }$ constructed from the dressing transformation.
As we show below, all $\om_{\vP,T}$ lie in the same sector and hence the corresponding GNS representations
are unitarily equivalent.  {\bc Thus} there is no velocity superselection in this situation. 
\bet \label{theorem:equivalence}
 Let $T$ be an infravacuum map. For all $\vP, \vP'\in \mcS$, we have $[\om_{\vP,T}]_{\mrm{In}(\mfa) }=[\om_{\vP',T}]_{\mrm{In}(\mfa) }$.
\eet
\proof Since for any {\bc  unit vector} $\Phi{\bc \in \mcF}$ 
we can find $i_{\Phi}\in \sect$ such that $\om_{\Phi}=\om_{\text{vac}}\circ i_{\Phi}$, \cite[Theorem 10.2.1]{KR}, we write
\beqa
\, [\om_{\vP,T}]_{\mrm{In}(\mfa)}=[\om_{\text{vac}}]_{\mrm{In}(\mfa)}\circ \al_T\circ \al_{\vv_{\vP}}=[\om_{\text{vac}}]_{\mrm{In}(\mfa)}\circ \al_{(T^{-1})^{t} \vv_{\vP}} \circ \al_T
=[\om_{\text{vac}}]_{\mrm{In}(\mfa)}\circ \al_T, \,\,\,\,\,\,
\eeqa
where in the second step we used the homomorphism property of $\al$ (cf.~(\ref{automorphism-definition})) and
the multiplication law in  $\mrm{ISp}(\mcL)$ which gives $T\cdot \vv_{\vP} \cdot T^{-1}=(T^{-1})^t \vv_{\vP}$.
In the third step we used Theorem~\ref{infravacuum-theorem} which ensures that $(T^{-1})^t \vv_{\vP}\in L^2_{\mrm{tr}}(\real^3;\complex^3)$
and thus  $\al_{(T^{-1})^{t} \vv_{\vP}} \in G_{[\om_{\text{vac}}]_{\mrm{In}(\mfa)}  }$ by the last statement of Proposition~\ref{sector-proposition}. \qed

As a consequence of Theorem~\ref{theorem:equivalence}, we have $\ov{[[\om_{\vP,T}]_{\mrm{In}(\mfa)}  ] }^{\al_T} = \ov{[  [\om_{\vP',T}]_{\mrm{In}(\mfa)}  ]   }^{\al_T}$ for all $\vP \neq \vP'$. 
In fact, since $[\om_{\vP,T}]_{\mrm{In}(\mfa)} $ is independent of $\vP$ by the above result, it is clear from Definition~\ref{conjugation-def}
 that $G^{\al_{T}}_{\omega_{\vP,T}, \omega_{\text{vac}}}$ is independent of $\vP$ as well, which implies the same for $\ov{[[\om_{\vP,T}]_{\mrm{In}(\mfa)}  ] }^{{\bc \al_T} } $.  

As the states (\ref{state-one}) rely on the non-canonical decomposition of the single-electron  state   {\bc  $\om_{\vP}$}  into the `undressed electron' {\bc $\om_{\Phi_{\vP}} $} and the `dressing' {\bc $\al_{\vv_{\vP}}$},
 their physical realization appears difficult, even as a thought experiment. { To improve on this, one could } consider the following family of states, which does not rely on such a decomposition:
 \begin{equation}
\om_{\vP}^{T} :={\bc \om_{\vP}\circ \al_T}=\om_{\Phi_{\vP}}\circ\al_{\vv_{\vP}}\circ \al_{T}. \label{state-two}
\end{equation}
However, velocity superselection \emph{persists} for these states, as shown below: The sectors $[\om^T_{\vP}]_{\mrm{In}(\mfa) }$ depend on $\vP$, and correspondingly, the GNS representations are not equivalent.  It is clear from the proof below that replacing $\al_T$ in (\ref{state-two}) with  any other automorphism of $\mfa$ will not improve the situation. Thus the modification (\ref{state-one}) is the only possibility we can see to cure velocity superselection at the level of states with the help of the infravacuum.
\bep  Let $T$ be an infravacuum map. Then, $[\om^T_{\vP}]_{\mrm{In}(\mfa) }\neq [\om^T_{\vP'}]_{\mrm{In}(\mfa) }$ for all  $\vP, \vP'\in \mcS$ such that $\vP\neq \vP'$.
\eep 
\proof
We argue by contradiction: Suppose $[\om^T_{\vP}]_{\mrm{In}(\mfa)}=[\om^T_{\vP'}]_{\mrm{In}(\mfa)}$ for some $\vP\neq \vP'$ from $\mcS$.
This means $[\om_{\vP}]_{\mrm{In}(\mfa)}\circ \al_T=[\om_{\vP'}]_{\mrm{In}(\mfa)}\circ \al_T$ in conflict with Proposition~\ref{velocity-superselection-x}. \qed

\section{Kraus-Polley-Reents infravacua  } \label{Map-T}
\setcounter{equation}{0}
 In this section we provide a proof of Proposition~\ref{theo:convergenceseq}.

First, we introduce the decomposition of functions from $\mcL=\bigcup_{\thetae>0} L^2_{\mrm{tr}, \thetae}(\real^3; \complex^3) $ into radial and angular parts
\beqa
L^2_{\mrm{tr}, \thetae}(\real^3; \complex^3)=L^2_{\thetae}(\real_+)\otimes L^2_{\mrm{tr}} (S^2;\complex^3). \label{radial-angular-decomposition}
\eeqa   
Here $L^2_{\thetae}(\real_+)$ is the space of radial functions with  measure $|\vk|^2d|\vk|$ vanishing for $|\vk|\leq \thetae$ and we write $L^2(\real_+):=L^2_{\thetae=0}(\real_+)$. Further,
$L^2_{\mrm{tr}} (S^2;\complex^3)$  is the space of  transverse angular functions with the natural spherical measure 
 $d\Om(\theta, \phi):=\sin\theta d\theta d\phi$. The latter space is the range of the projection
 $P^{\mrm{tr}}: L^2(S^2; \complex^3) \to L^2(S^2;\complex^3)$ given by
\beqa
(P^{\mrm{tr}}\vf)(\hat{\vk}):=\sum_{\la=\pm }  \big(\vf(\hat{\vk})\cdot  \pmb{\epsilon}_\la (\hat{\pmb{k}}) \big) \pmb{\epsilon}_\la (\hat{\pmb{k}}). \label{upper-transverse}
\eeqa
As in \cite{KPR77}, we introduce the transverse, vector valued spherical harmonics  for $\ell \neq 0$, $-\ell \leq m\leq \ell$,
\begin{equation}\label{vSH}
\pmb{Y}_{\ell m \pm} = \frac{1}{\sqrt{\ell(\ell +1)}}\pmb{a}_\pm Y_{\ell m} \quad \text{with} 
\quad \pmb{a}_+ = |\pmb{k}| \nabla_{\pmb{k}} \quad \text{and} \quad \pmb{a}_- = \hat{\pmb{k}} \times \pmb{a}_+
\end{equation}
which are elements of $L^2_{\mrm{tr}} (S^2; \complex^3)$. Here $Y_{\ell m}$ are the usual spherical harmonics  which form an orthonormal basis in $L^2(S^2)$ with measure $d\Om$. As a consequence,
 $\pmb{Y}_{\ell m\pm}$ with $\ell \neq 0$ form a complete basis of $L^2_{\mrm{tr}} (S^2; \complex^3)$.
Furthermore, $\vJ^2 \pmb{Y}_{\ell m\pm}=\ell  (\ell+1) \pmb{Y}_{\ell m\pm}$,
where $\vJ=\vL+\vS$ is the total angular momentum of photons including the orbital and spin part  (the latter being a constant matrix on $\complex^3$).  These properties of $\pmb{Y}_{\ell m\pm}$ are discussed more extensively in Appendix~\ref{appendix:vectorSH}. 

Keeping the decomposition~(\ref{radial-angular-decomposition}) in mind, we define the relevant symplectic map~$T$. 
\bed\label{infravacuum}  The  Kraus-Polley-Reents infravacuum  map $T: \mcL\to \mcL$ is defined as follows\footnote{By small modifications of this definition one can easily obtain many different maps satisfying properties (a), (b) from Proposition~\ref{theo:convergenceseq}. We restrict attention to the simplest choice.}:   

\begin{itemize}

\item We introduce  sequences $\eps_i:=2^{-(i-1)}\ka$ and $b_i:= \fr{1}{i}$ for $i=1,2,3\ldots$.

\item We define functions $\xi_i(|\vk|):=\fr{\chi_{[\eps_{i+1},\eps_i] }(|\vk|)}{|\vk|^{3/2}}\in L^2(\real_+)$ and their normalized  counterparts $\ti\xi_i(|\vk|):= \xi_i(|\vk|) / \|\xi_i\|_{L^2(\real_+)}$. 

\item 
We define the orthogonal projections $\pmb{Q}_i: L^2_{\tr}(\real^3;\complex^3)  \to L^2_{\tr}(\real^3;\complex^3)$
and $\tilde{\pmb{Q}}_i: L^2_{\tr}(S^2;\complex^3)\to L^2_{\tr}(S^2; \complex^3)$ given by
\begin{equation}\label{pmbQi}
\pmb{Q}_i = |\ti\xi_i\ran\lan \ti \xi_i| \otimes \tilde{\pmb{Q}}_i \quad \text{with} \quad \tilde{\pmb{Q}}_i := \sum_{0< \ell \leq i} \sum_{m=-\ell}^\ell \sum_{\lambda = \pm} | \pmb{Y}_{\ell m \lambda} \rangle \langle \pmb{Y}_{\ell m \lambda} |.
\end{equation}
Since $ |\ti\xi_i\ran\lan \ti \xi_i|: L^2_{\thetae}(\real_+)\to L^2_{\eps_{i+1}}(\real_+)$ we have $\pmb{Q}_i: \mcL\to \mcL$.

\item We introduce the complex-linear maps $T_1, T_2: \mcL\to\mcL$ 
\begin{equation}\label{T1T2}
 T_1 := \pmb{1} + \operatorname*{s-lim}_{n \to \infty} \sum_{i =1}^n (b_i -1)\pmb{Q}_i, \quad 
 T_2 := \pmb{1} +\operatorname*{s-lim}_{n \to \infty} \sum_{i=1}^n \big(\frac{1}{b_i} - 1 \big)\pmb{Q}_i.
\end{equation}
Clearly, the maps are well-defined, since $\pmb{Q}_i\vf=0$ for any $\vf\in\mcL$ and sufficiently large $i$. We will denote by $T_{1,n}, T_{2,n}$ the respective approximants.

\item Let $\hGa$ be the complex conjugation in momentum space. We introduce the real-linear map $T: \mcL\to \mcL$ given by
\beqa
T:=T_1\fr{1+\hGa}{2}+T_2\fr{1-\hGa}{2},
\eeqa
and denote by $T_n$ the respective approximants. In other words, writing $\vf=\vf_1+i \vf_2\in \mcL$, where $\vf_1,\vf_2$ are real-valued,
we have $T\vf=T_1\vf_1+iT_2\vf_2$.
\end{itemize}
\eed

\nin  Following  \cite{KPR77, Ku98}, we show that  $T$ is symplectic and invertible.  We also verify that $T_1$ extends to a bounded map
on  $L^2_{\mrm{tr}}(\real^3; \complex^3 )$ which yields Proposition~\ref{theo:convergenceseq} (b). 
\bel\label{T-prop} We have for  $\vf_0,\vf_1,\vf_2\in \mcL$
\beqa
& &\lan T_{1}\vf_1, T_{2}\vf_2\ran=\lan \vf_1, \vf_2\ran, \quad T_1T_2\vf_0=\vf_0, \quad T_2T_1\vf_0=\vf_0. \label{T-properties}
\eeqa
Furthermore, $T$ is symplectic and invertible with $T^{-1}= T_2\fr{1+\hGa}{2}+T_1\fr{1-\hGa}{2}$  and $T_1$ extends to a bounded operator on $L^2_{\mrm{tr}}(\real^3;\complex^3)$.   
\eel
\proof Making use of self-adjointness of $T_{1,n}, T_{2,n}$ on $L^2_{\tr}(\real^3;\complex^3)$ and of the fact that $\pmb{Q}_i$ are mutually orthogonal 
projections, we have
\beqa
\lan T_{1,n}\vf_1, T_{2,n}\vf_2\ran  
= \lan \vf_1,  (\pmb{1} +  \sum_{i' =1}^n (b_{i'} -1)\pmb{Q}_{i'}) ( \pmb{1}+ \sum_{i=1}^n \big(\frac{1}{b_i} - 1 \big)\pmb{Q}_i) \vf_2 )\ran=\lan \vf_1, \vf_2\ran.\,\,\,\,\,
  \eeqa
This computation gives   properties~(\ref{T-properties}). Now choosing $\vf,\vf'\in \mcL$ and decomposing them into real and imaginary parts as $\vf=\vf_1+i\vf_2$, $\vf'=\vf'_1+i\vf'_2$,
we have
\beqa
\mrm{Im}\lan T_1 \vf_1+iT_2\vf_2,  T_1 \vf'_1+iT_2\vf'_2\ran=-\lan \vf_2, \vf_1'\ran+ \lan \vf_1, \vf_2'\ran= \mrm{Im}\lan \vf, \vf'\ran.
\eeqa
Next, we set $\hat{T}:= T_2\fr{1+\hGa}{2}+T_1\fr{1-\hGa}{2} $ and check $T\hat{T}\vf=\hat{T}T\vf=\vf$ using (\ref{T-properties}). 
 Finally, the boundedness of $T_1$ follows from the computation 
\beqa
\|\sum_{i=1}^{\infty} (b_i-1)\pmb{Q}_i\vf\|^2\leq C\sum_{i=1}^n\lan \vf, \pmb{Q}_i \vf\ran\leq C\|\vf\|^2,
\eeqa
where we could choose a finite $n$ in the first step since $\vf\in \mcL$.  As $C$ is independent of $n$,
this concludes the proof. \qed 

 In preparation for the proof of  Proposition~\ref{theo:convergenceseq} (a), we set $\si_n=\varepsilon_n$, $n\in\nat$, where $\varepsilon_n$ appeared in Definition~\ref{infravacuum}.
Recalling definition (\ref{v-P-sigma}) and setting as above $\xi_i(|\vk|):=|\vk|^{-3/2} \chi_{[\eps_{i+1},\eps_i] }(|\vk|)$, we write 
\beqa
\vv_{\vP,\si_n}(\vk)=\sum_{i=1}^{n-1}(\xi_i(|\vk|) \otimes \vvp_{\vP,n}^{\tr}(\hat{\vk})), \label{definitions-of-vectors}
\eeqa
where
\beqa
\vvp_{\pmb{P},n}^{\tr}(\hat{\vk}) :=  P^{\tr}\vvp_{\pmb{P},n}(\hat{\vk}), \quad \vvp_{\pmb{P},n}(\hat{\vk}) :={\bc \ti{\al}^{1/2} }\frac{ \nabla E_{\pmb{P}, \sigma_n}  } {1- \hat{\vk}\cdot \nabla E_{\vP,\si_n}}.
\eeqa 
We also define 
\beqa
\vvp_{\pmb{P}}^{\tr}(\hat{\vk}):=  P^{\tr}\vvp_{\pmb{P}}(\hat{\vk}), \quad \vvp_{\pmb{P}}(\hat{\vk}):= {\bc \ti{\al}^{1/2} } \frac{\nabla E_{\pmb{P}}  } {1- \hat{\vk}\cdot \nabla E_{\vP}}.  \label{future-reference}
\eeqa
As a consequence of Lemma~\ref{spectral} and dominated convergence we have
\beqa
\lim_{n\to\infty}\|  \vvp_{\pmb{P},n}-  \vvp_{\pmb{P}}\|_{L^2(S^2;\complex^3)}=0 \label{pointlike-convergence}
\eeqa
from which it also follows that  $\lim_{n\to\infty}\|  \vvp^{\tr}_{\pmb{P},n}-\vvp^{\tr}_{\pmb{P}}\|_{L^2_{\mrm{tr}}(S^2;\complex^3)}=0$ since $\| P^{\tr}\| =1$.

Since $\vv_{\vP,\si_n}\in \mcL$ are real-valued, we have $T\vv_{\vP,\si_n}=T_1 \vv_{\vP,\si_n}$. 
By definition of $T_1$ and $\pmb{Q}_j${\bc, } and  by the support properties of $\xi_i$, we have
\begin{align}
T_1 \vv_{\vP,\si_n} &= (\pmb{1} +  \sum_{j =1}^{\infty} (b_j-1)\pmb{Q}_j) \vv_{\vP,\si_n}
 =\sum_{i=1}^{n-1} (\pmb{1} + (b_i -1) \pmb{Q}_i) (\xi_i\otimes  \vvp_{\vP,n}^{\tr}).
\end{align}
We set $\psi_{i,n}:=(  \pmb{1} + (b_i -1) \pmb{Q}_i  ) (\xi_i\otimes  \vvp_{\vP,n}^{\tr})$ and note that $\lan \psi_{i,n}, \psi_{i',n'}\ran_{ L^2_{\tr}(\real^3;\complex^3)}=0$ for $i\neq i'$
and any $n,n'\in\nat$. We also observe that by (\ref{pointlike-convergence})
\beqa\label{lim:psi}
\psi_i:=\lim_{n\to \infty} \psi_{i,n}= (  \pmb{1} + (b_i -1) \pmb{Q}_i  ) (\xi_i\otimes  \vvp^{\mrm{tr} }_{\vP})
\eeqa
exists in $L^2_{\tr}(\real^3;\complex^3)$. 
We want to show that the vector  $T_1\vv_{\vP}:= \sum_{i=1}^{\infty}\psi_{i}$ exists in $L^2_{\tr}(\real^3;\complex^3)$, as it is a natural candidate for $\lim_{n\to\infty} T_1\vv_{\vP,\si_n}$.
For this, we prove the following lemma.
\begin{lemma}\label{lemma:psi}
{\bc Fix $i\in \nat$. Then} there exists $C_i>0$ independent of $n$ such that
\beqa
\sup_{n\in\nat} \| \psi_{i,n}\|_{ L^2_{\tr}(\real^3;\complex^3) }^2\leq C_i, \quad \text{where} \quad \sum_{i=1}^{\infty}C_i<\infty. \label{key-relation}
\eeqa
In particular, {\bc $\sum_{i=1}^\infty \psi_{i,n}$ and $\sum_{i=1}^\infty \psi_{i}$ exist in $L^2_{\tr}(\real^3;\complex^3)$, and}  $\lim_{n \to \infty}\sum_{i=1}^\infty \psi_{i,n}=\sum_{i=1}^\infty \psi_i$ {\bc in the same topology.}
\end{lemma}

\begin{proof}
To verify (\ref{key-relation}) we estimate the norm 
\beqa
\|\psi_{i,n}\|^2_{ L^2_{\tr}(\real^3;\complex^3) }\y&=&\y\|\xi_i\otimes \big(\pmb{1}+(b_i -1) \ti{\pmb{Q}}_i \big)\vvp_{\pmb{P},n}^{\tr}\|_{ L^2_{\tr}(\real^3;\complex^3) }^2\non\\
\y &=&\y \|\xi_i\|^2_{L^2(\real_+)} \| (\pmb{1}- \ti{\pmb{Q}}_i) \vvp_{\pmb{P},n}^{\tr} + b_i \ti{\pmb{Q}}_i\vvp_{\pmb{P},n}^{\tr}\|^2_{L^2_{\mrm{tr}}(S^2;\complex^3)}\non\\
\y&\leq &\y  \ln \fr{\epsilon_i}{\epsilon_{i+1}}\big( \| (\pmb{1}- \ti{\pmb{Q}}_i) \vvp_{\pmb{P},n}^{\tr}\|^2_{L^2_{\mrm{tr}}(S^2;\complex^3)}+ b_i^2   \|\vvp_{\pmb{P},n}^{\tr}\|^2_{L^2_{\mrm{tr}}(S^2;\complex^3)} \big). \label{formula-for-squares}
\eeqa

For the last term, \eqref{pointlike-convergence} gives $\|\vvp^{\tr}_{\pmb{P},n}\|^2_{L^2_{\mrm{tr}}(S^2;\complex^3)}\leq c_1$ uniformly in $n$.

Concerning the term with $(\pmb{1}- \ti{\pmb{Q}}_i)$, we make use of the completeness of the vector-valued spherical harmonics:
\beqa
\| (\pmb{1}- \ti{\pmb{Q}}_i) \vvp_{\pmb{P},n}^{\tr}\|^2_{L^2_{\mrm{tr}}(S^2;\complex^3)}
\y&=&\y\sum_{\ell>i} \sum_{m=-\ell}^{\ell}\sum_{\la=\pm} |\lan \pmb{Y}_{\ell  m\la}, \vvp_{\pmb{P},n}\ran_{L^2(S^2;\complex^3) }|^2\non\\
\y& \leq &\y
 \fr{2 }{i^2}\big( \lan  \vvp_{\pmb{P},n}, \vL^2 \vvp_{\pmb{P},n}\ran_{L^2(S^2;\complex^3)} + 2\|  \vvp_{\pmb{P},n} \|^2_{L^2(S^2;\complex^3)} \big), 
\label{spherical-harmonics-expression}
\eeqa
where $\vvp_{\pmb{P},n}$ is defined in (\ref{future-reference}) and we replaced $\vvp_{\pmb{P},n}^{\tr}$ with $\vvp_{\pmb{P},n}$, which is justified since  $\pmb{Y}_{\ell m\la}$ are transverse.
In the  last step we exploited  $\vJ^2  \pmb{Y}_{\ell m\pm}=\ell(\ell+1) \pmb{Y}_{\ell m\pm}$, inserted $\vJ^2= (\vL+ \vS)^2$,  applied the Cauchy-Schwarz inequality  to terms of the form $L_i S_i$ with $i=1,2,3$, and   used $\| \vS^2 \| = 2$.
Here $\|\vvp_{\pmb{P},n} \|^2_{L^2(S^2;\complex^3)}$ is uniformly bounded in $n$ due to \eqref{pointlike-convergence}.
To estimate the first term in (\ref{spherical-harmonics-expression}) we write 
$\pmb{L}^2 = -  \frac{1}{\sin\theta} \frac{\partial}{\partial \theta}\Big(\sin\theta \frac{\partial}{\partial \theta}  \Big) - \frac{1}{\sin^2\theta}\frac{\partial^2}{\partial \phi^2} $ in spherical coordinates, choosing our reference frame in $\real^3$ such that $\nabla E_{\vP,\si_n}$ is in the direction 
of the third axis, and compute 
\begin{multline}
\lan \vvp_{\pmb{P},n}, \vL^2 \vvp_{\pmb{P},n}\ran_{L^2(S^2;\complex^3)}\\
=-2\pi {\bc \ti \al} |\nabla E_{\pmb{P}, \sigma_n}|^2 \int_{-1}^{1} \frac{  d\cos\,\theta  } {1- \cos\theta |\nabla E_{\vP,\si_n}|}  \frac{\partial}{\partial \cos\theta}\Big(\sin^2\theta \frac{\partial}{\partial \cos\theta}\Big)  \frac{1} {1- \cos\theta |\nabla E_{\vP,\si_n}|}\\
=2\pi {\bc \ti \al} |\nabla E_{\pmb{P}, \sigma_n}|^2 \int_{-1}^{1}  dt\,(1-t^2)\Big(\fr{d}{dt} \fr{1}{1-t|\nabla E_{\vP,\si_n}| } \Big)^2\leq c_2,
\end{multline}
where we introduced the variable $t:=\cos\theta$. 
This bound is uniform in $n$ by Lemma~\ref{spectral}. Coming back to formula~(\ref{formula-for-squares}) and collecting our estimates we obtain
\beqa
\|\psi_{i,n}\|^2_{ L^2_{\tr}(\real^3;\complex^3) }\leq  C\ln\,\fr{\epsilon_i}{\epsilon_{i+1}}\bigg( \fr{1}{i^2}+ b_i^2\bigg),
\eeqa
where $C$ is independent of $n$. Given the choice of $\eps_i, b_i$, $i\in\nat$, in Definition~\ref{infravacuum}, estimates~(\ref{key-relation}) follow. Using this and \eqref{lim:psi}, the last statement of the lemma follows by dominated convergence, noting that the $\psi_{i,n}$ are mutually orthogonal.
\end{proof}

\nin Using these results we are now ready to show convergence of $T \vv_{\vP,\si_n}$.\\
\nin \emph{Proof of Proposition~\ref{theo:convergenceseq} (a)}.
We write
\beqa
T_1 \vv_{\vP,\si_n}-T_1\vv_{\vP}=\sum_{i=1}^{\infty}  (\psi_{i,n}-\psi_i)-\sum_{i= n}^{\infty}  \psi_{i,n}.
\eeqa
In the limit $n \to \infty$ we apply Lemma~\ref{lemma:psi} to the first sum. As for the second term, since
\beqa
\Big\| \sum_{i= n}^\infty \psi_{i,n} \Big\|_{ L^2_{\tr}(\real^3;\complex^3) }  ^2  = \sum_{i= n }^\infty \| \psi_{i,n} \|^2_{ L^2_{\tr}(\real^3;\complex^3) },
\eeqa
where $\| \psi_{i,n} \|^2_{ L^2_{\tr}(\real^3;\complex^3) }\leq C_i$, see (\ref{key-relation}), it is the reminder term of a convergent series, and therefore vanishes for $n \to \infty$.
This concludes the proof. \qed


\section{Proof of velocity superselection  } \label{proof-of-velocity-superselection}
\setcounter{equation}{0}

The goal of this section is to provide a proof of Proposition~\ref{velocity-superselection-x}.

Suppose, by contradiction, that $[\om_{\vP}]_{\sect}=[\om_{\vP'}]_{\sect}$ for some $\vP\neq \vP'$ from $\mcS$. That is,
$\om_{\vP}=\om_{\vP'}\circ i_{\vP, \vP'}$ for some $i_{\vP, \vP'}\in \sect$. Furthermore, using formula~(\ref{state-formula}),
we can write
\beqa
\om_{\Phi_{\vP}}\circ \al_{\vv_{\vP}}= \om_{\Phi_{\vP'}}\circ \al_{\vv_{\vP'}}\circ  i_{\vP, \vP'}, \label{inner-formula}
\eeqa
where, as before, $\om_{\Phi}(\,\cdot\,)=\lan \Phi, \pi_{{\text{vac}}}(\,\cdot \,) \Phi\ran$. Again, {\bc for any unit vector $\Phi\in \mcF$} we can find $i_{\Phi}\in \sect$ such that $\om_{\Phi}=\om_{\text{vac}}\circ i_{\Phi}$. 
{\bc Hence}, we obtain from (\ref{inner-formula}) that
\beqa
\om_{\text{vac}}\circ \al_{\vv_{\vP}}= \om_{\text{vac}}\circ  \ti{i}_{\vP, \vP'}\circ \al_{\vv_{\vP'}} \label{equality-to-disprove}
\eeqa
for another $\ti{i}_{\vP, \vP'}\in \sect$.
To disprove this equality, we choose  some $\vg\in \mcL$, purely imaginary and integrable, set $\vg_{s}(\vk):=s^{3/2}\vg(s \vk)$,  $s>0$,
and consider the sequence $s\mapsto  W(\vg_{s})$.  We will evaluate this sequence  on the two states appearing  in (\ref{equality-to-disprove}).

First, using formula~(\ref{automorphism-definition}), we get
$\al_{ \vv_{ \pmb{P} } }(W(\vg_{s}))= e^{-2i\mrm{Im} \lan \vv_{\pmb{P}}, \vg_{s} \ran} W(\vg_{s})$. With the help of~(\ref{v-P}), 
the fact that $\vg$ is transverse, and the dominated convergence theorem, we get
\beqa
\lim_{s\to\infty}\lan \vv_{\pmb{P}}, \vg_s\ran
\y&=&\y\lim_{s\to\infty} \alt^{\bc 1/2} \int d^3k\,  \frac{\chi_{[0, \kappa]} (|\pmb{k}|/s)}{|\pmb{k}|^{3/2}}\frac{ \nabla E_{ \pmb{P} }  }{1  - \hat{\pmb{k}}\cdot \nabla E_{\pmb{P}}  } \cdot\vg( \vk)\non\\
& &\ph{444444}=  \alt^{\bc 1/2} \int d^3k\, \frac{1}{|\pmb{k}|^{3/2}}\frac{ \nabla E_{\pmb{P} } }{1  - \hat{\pmb{k}}\cdot \nabla E_{\pmb{P}} }  \cdot  \vg( \vk),
\label{scalar-product-formula}
\eeqa
 and denote the last expression by $\lan \vv_{\pmb{P}}, \vg_{\infty}\ran$.  Since  $s\mapsto  W(\vg_{s})$ is a central sequence (see Lemma~\ref{central-sequence} below),  
equality~(\ref{equality-to-disprove}) and definition (\ref{vacuum-state-def}) imply
\beqa
e^{-2i \mrm{Im}\lan \vv_{\pmb{P}}, \vg_{\infty}\ran}=e^{-2i \mrm{Im}\lan \vv_{\pmb{P'}}, \vg_{\infty}\ran}. \label{exponential-functions}
\eeqa

Now we want to achieve a contradiction by choosing $\vg$ so that the above equality fails. For this purpose we first set $\vg=i\vg_1$, where $\vg_1$ is real.
Recalling that $\vv_{\pmb{P}}$ is real, we use (\ref{scalar-product-formula}) to write 
\beqa
\mrm{Im}\lan \vv_{\pmb{P}}, \vg_{\infty}\ran-\mrm{Im}\lan \vv_{\pmb{P'}}, \vg_{\infty}\ran
\y &=&\y  \ti\alpha^{\bc 1/2} \int d^3k\, \bigg( \frac{ \nabla E_{\vP} } {1  - \hat{\vk}\cdot \nabla E_{\vP}} -  \frac{ \nabla E_{\vP'}} {1  - \hat{\vk}\cdot  \nabla E_{ \vP'} }  \bigg) \! \cdot \! \fr{\vg_1( \vk)}{|\vk|^{3/2}}\non \\
\y &=&\y  \ti\alpha^{\bc 1/2}  \int d^3k\, \bigg( \frac{ \nabla E_{\vP} -\hat{\vk}  } {1  - \hat{\vk}\cdot \nabla E_{\vP}} -  \frac{ \nabla E_{\vP'} -\hat{\vk} } {1  - \hat{\vk}\cdot  \nabla E_{ \vP'} }  \bigg) \! \cdot \! \fr{\vg_1( \vk)}{|\vk|^{3/2}}, \,\,\,\,\, \,\,\,\,\,\,\,\,\,\, \label{bracket}
\eeqa
where in the second step we used that $\vg_1$ is transverse to take the transverse part of the expression in bracket. We denote the
expression in bracket in (\ref{bracket}) by $F_{\vP,\vP'}(\hat{\vk})$ 
and define $\vg_1$ as follows:
\beqa
\vg_1(\vk):=C\fr{F_{\vP,\vP'}(\hat{\vk})} {| F_{\vP,\vP'}(\hat{\vk})|^2 } \chi_{\vP,\vP'}(\hat\vk)\chi_{[\si, \ka]}(|\vk|), \label{vg1}
\eeqa
where we chose some $C\neq 0$,  $0<\si<\ka$  and $\chi_{\vP, \vP'}$ is a non-zero, bounded, positive function from $L^2(S^2)$ which vanishes near 
$\mrm{Span}\{\nabla E_{\vP},\nabla E_{\vP'}\}$.  We note that  by Lemma~\ref{spectral} (a) we have $\nabla E_{\vP}\neq \nabla E_{\vP'}$ for $\vP\neq \vP'$. 

Now it follows from 
Lemma~\ref{geometry} below that  the denominator in (\ref{vg1}) is non-zero on the support of $ \chi_{\vP,\vP'}$. With this
definition, (\ref{bracket}) gives
\beqa
\mrm{Im}\lan \vv_{\pmb{P}}, \vg_{\infty}\ran-\mrm{Im}\lan \vv_{\pmb{P'}}, \vg_{\infty}\ran=
C   \alt^{\bc 1/2} \int d^3k\,  \chi_{\vP,\vP'}(\hat\vk)\fr{\chi_{[\si, \ka]}(|\vk|) }{|\vk|^{3/2}},
\eeqa
which is manifestly non-zero. By varying $C$ we can avoid equality in (\ref{exponential-functions}) due to the periodicity of the exponential function. This 
concludes the proof of Proposition~\ref{velocity-superselection-x}. 
\bel  \label{central-sequence} For $\vg\in \mcL\cap L^1_{\mrm{tr}}(\real^3;\complex^3)$ and $s>0$  we set $\vg_{s}(\vk):=s^{3/2}\vg(s \vk)$.
Then  $s\mapsto   W(\vg_{s})$ is a central sequence, i.e.,
\beqa
\lim_{s\to \infty}\|[A, W(\vg_{s})]\|=0 \textrm{ for any } A\in\mfa.
\eeqa
\eel
\proof  For any given $\vf\in \mcL$ we compute the commutator
\beqa
\,[W(\vf), W(\vg_s)]\y&=&
 \big(e^{-i\mrm{Im}\lan \vf, \vg_s\ran }-   e^{i\mrm{Im}\lan \vf, \vg_s\ran } \big)W(\vf+\vg_s).\label{commutator-formula-superselection-proof}
\eeqa
Now we  find $\vf_n\in \mcL$ which are bounded functions\footnote{i.e., each $\vk\mapsto |\vf_n(\vk)|$ is bounded, where $|\,\cdot \,|$ means norm in $\complex^3$.} 
and such that $\|\vf-\vf_n\|\to 0$ as $n\to\infty$. The expression in the exponential in the bracket above has the form
\beqa
\,\lan \vf, \vg_s\ran=\fr{1}{s^{3/2}} \int d^3k\, \bar{\vf_n}(\vk/s)\cdot \vg(\vk) + \lan \vf-\vf_n, \vg_s\ran.
\eeqa
Using that $\vg\in L^1_{\mrm{tr}}(\real^3;\complex^3)$ and each $\vf_n$ is bounded, we can apply the dominated convergence theorem to the integral above for any fixed $n$.
Exploiting in addition that  $\|\vg_s\|=\|\vg\|$, we get  $\lim_{s\to\infty}\,\lan \vf, \vg_s\ran=0$.  Consequently, (\ref{commutator-formula-superselection-proof}) gives
\beqa
\lim_{s\to \infty}\|[W(\vf), W(\vg_s)]\|=0 \textrm{ for any } \vf\in \mcL. \label{f0}
\eeqa
Since any $A\in \mfa$ can be approximated in norm by finite linear combinations of $W(\vf)$, $\vf \in \mcL$, the proof is complete. \qed

\bel\label{geometry} For any $\vv\in \real^3$, $|\vv|<1$, consider the function on $S^2$
\beqa
F_{\vv}(\hat{\vk}):=  \fr{\vv-\hat{\vk} } {1  - \hat{\vk} \cdot \vv}.
\eeqa
Suppose that $\vv_1, \vv_2\in\real^3$ are such that $\vv_1\neq \vv_2$. Then, for all $\hat\vk\in S^2$ such that $\hat\vk\notin\mrm{Span}\{ \vv_1, \vv_2\}$, we have 
\beqa
F_{\vv_1}(\hat\vk)\neq F_{\vv_2}(\hat\vk).
\eeqa
\eel

\proof Assume $\vv_1 \neq \vv_2$ and suppose that
\begin{equation}\label{Fproof}
\fr{\vv_1-\hat{\vk} } {1  - \hat{\vk} \cdot \vv_1}= \fr{\vv_2-\hat{\vk} } {1  - \hat{\vk} \cdot \vv_2}.
\end{equation}
Then $\hat{\vk} \cdot \vv_1 \neq \hat{\vk} \cdot \vv_2$. (If not, then \eqref{Fproof} implies $\vv_1 = \vv_2$ which is a contradiction.)

From \eqref{Fproof} we have
\begin{equation}
 \hat{\vk} = \fr{1  - \hat{\vk} \cdot \vv_2}{\hat{\vk} \cdot \vv_1 - \hat{\vk} \cdot \vv_2}\vv_1 - \fr{1  - \hat{\vk} \cdot \vv_1}{\hat{\vk} \cdot \vv_1 - \hat{\vk} \cdot \vv_2}\vv_2.
\end{equation}
This implies $\hat\vk\in\mrm{Span}\{ \vv_1, \vv_2\}$. \qed

\section{Conclusions}
\setcounter{equation}{0}

In this paper we introduced {\bcc the} concept of a relative normalizer for two subgroups $R\subset S$
of a group $G$ {\bc given by $N_G(R,S):=\{\, g\in G\,|\, g\cdot S\cdot g^{-1}\subset R\,\}$}.  We could show that the inhomogeneous symplectic group and the group of automorphisms
of the corresponding CCR algebra admit non-trivial relative normalizers. They are given by the infravacuum 
maps of Kraus, Polley and Reents and the corresponding Bogolubov transformations. 
Moreover, we studied the impact of
such relative normalizers on the superselection theory of this CCR algebra. We  gave canonical definitions of
the conjugate and second conjugate class of a given superselection sector with respect to a  reference `vacuum' state and 
a `background' automorphism. Then we showed that distinct sectors may give rise to coinciding conjugate and second conjugate
classes if they are computed with respect to the infravacuum background. This shed a new light on the problem of velocity superselection 
in non-relativistic QED.

 Our findings warrant further investigations of  infravacuum representations in  QED. 
 One question which is left open is how large  the second conjugate classes of states of physical interest are.
 It is easy to conclude from 
 Lemma~\ref{conjugation-lemma}  and Proposition~\ref{velocity-superselection-x}  that
 $\ov{\ov{  [[\om_{\vac}]_{\mrm{In}(\mfa)}]}}^{\al_{T}}\neq  \ov{\ov{  [[\om_{\vac}\circ \al_T]_{\mrm{In}(\mfa)}]}}^{\al_{T}}$,
 thus in general second conjugate classes are proper subsets of the set of pure states.  
 We showed in Theorem~\ref{disjointness-proposition-alg} (c)  that second conjugate classes of physical states contain their usual
 soft-photon dressing, given by singular coherent representations. However,  we do not know if
 the inclusion in this latter result is an equality or not. As many other examples of
 soft-photon clouds can be found in the literature (see, e.g., \cite{BD84}) 
 it is not obvious that the equality should hold. 
 
 Another interesting direction is a resolution of velocity superselection at the level of states, following up
 on Theorem~\ref{theorem:equivalence}. The non-canonical character of the decomposition (\ref{state-one})
 could probably be overcome by forming suitable equivalence classes of all possible decompositions into
  the `undressed electron' and the `dressing'. However, we believe that there is a more natural approach to
  this problem which is less dependent on the infravacuum maps $T$: This is to choose the symplectic space $\mcL$
  in such a way that the algebra $\mfa$ has a local structure and consider the restrictions of $\om_{\vP}$ to
  the subalgebra $\mfa(V_+)\subset \mfa$ of the electromagnetic fields localized in the future lightcone $V_+$.
  We remark that such a choice of $\mcL$ is by no means prevented by the non-relativistic character of the
  considered model.  We conjecture that the (highly reducible) GNS~representations of $\om_{\vP} \restriction \mfa(V_+)$ 
  are  unitarily equivalent for different $\vP$ and that  they are closely related to
  the second conjugate classes $\ov{\ov{  [[\om_{\vP}]_{\mrm{In}(\mfa)}]}}^{\al_{T}}$ considered  in the present work.
  This research direction, which will be pursued elsewhere, aims at a verification in a   
 concrete model of the general ideas from \cite{BR14, AD15}.

 Yet another important direction is to address the problem of the  sharp mass of the electron in non-relativistic QED.
For this purpose one needs to transform the model Hamiltonian to a {\bcc Kraus-Polley-Reents} representation in 
a suitable way. We remark in this context that the formal expression  $\al_{T^{-1}}(H_{\pmb{P} })$
{\bc appears to have} infrared-regularized interaction terms, hinting at a possible presence of ground states.  On the other hand, 
$\al_{T^{-1}}$ acts also on $H_{\pho}$ and $P_{\pho}$  which enter in the formula for $H_{\pmb{P} }$.
The resulting modification seems in conflict with the existence of ground states, even in the absence of interaction. 
Thus we believe that a naive application of $\al_{T^{-1}}$ as a dressing transformation at time zero
will not yield a sharp mass of the electron yet. It may help to express $\al_{T^{-1}}$ in terms of asymptotic
(incoming and outgoing) fields and define accordingly two infravacuum Hamiltonians. This strategy
 is consistent with recent works in the setting of algebraic QFT 
which stress the role of the arrow of time for curing the infrared problems \cite{BR14, AD15}.

On a more speculative side,  we think the proposal of Hawking, Perry and Strominger \cite{HPS16} concerning the
black-hole information paradox should be reconsidered from the infravacuum perspective. If the 
relevant infrared degrees of freedom of the gravitational field  can be blurred by  infravacuum-type  radiation of arbitrarily low energy,    
can they really encode information about the history of the black-hole formation? We hope to come back to the
above questions in  future publications.

\begin{appendix}

\section{Some auxiliary lemmas about $C^*$-algebras}

For the reader's convenience we recall some standard facts from the theory of $C^*$-algebras
which we use in our paper. We refer to \cite[Chapter 10]{KR}  for a more extensive discussion.
\bel\label{Lemma-A-two}  Let $(\pi,\hil)$ be an irreducible representation of a $C^*$-algebra $\mfa$ and let
$\Psi\in \hil$ be a unit vector. Then the GNS representation $(\ti{\pi}, \ti\hil, \ti\Om)$
induced by the state
\beqa
\om(A)=\lan \Psi, \pi(A)\Psi\ran, \quad A\in \mfa\bc{,} 
\eeqa
is unitarily equivalent to $(\pi,\hil)$.
\eel
\proof  The map $V: \hil \to \ti{\hil}$, given by $V\pi(A)\Psi=\ti{\pi}(A){\bc \ti\Om}$, $A\in \mfa$,
is densely defined by irreducibility of $\pi$ and has a dense range by cyclicity of ${\bc \ti{\Om}}$. 
By a straightforward computation one checks that $V$ is an isometry and $V\pi(A)=\ti\pi(A)V$
for all $A\in \mfa$. \qed\\
\nin In the following proposition we write $[\om]:=[\om ]_{\mrm{In}(\mfa)}$ for brevity. 
\bep\label{sector-proposition} Let $\om$ be a pure state on $\mfa$ and $(\pi, \hil, \Om)$ its GNS representation.
Then $[\om]$ coincides with  the set of all states whose GNS representations are unitarily equivalent to~$\pi$.
Furthermore,    $[\om\circ \ga]=[\om]$ iff  $\ga\in \Aut(\mfa)$ is unitarily implementable in $\pi$. 
\eep  
\proof Let $\om_1\in [\om]$ and denote by $(\pi_1, \hil_1, \Om_1)$ its GNS representation.  Since $\om_1=\om\circ \mrm{Ad}U$ for some unitary $U\in \mfa$,
 we can write
\beqa
& &\om(A)=\lan\Om, \pi(A)\Om\ran, \, \\
& &\om_1(A)=\lan \Om_1, \pi_1(A)\Om_1\ran=\lan \Om, \pi(UAU^*)\Om\ran=\lan \pi(U^*)\Om, \pi(A) \pi(U^*)\Om\ran,
\eeqa
for $A\in \mfa$. Now Lemma~\ref{Lemma-A-two} gives unitary equivalence of $\pi$ and $\pi_1$.

Conversely, suppose that $\pi$ and $\pi_1$ are unitarily equivalent, i.e., $\pi_1(A)=V\pi(A)V^*$ for all $A\in \mfa$ and some
unitary $V: \hil\to \hil_1$. Thus we can write
\beqa
\om_1(A)=\lan \Om_1, V\pi(A)V^* \Om_1\ran=\lan \Om, \pi(U)\pi(A)\pi(U)^*\Om\ran=\om(UAU^*),
\eeqa
where we used the irreducibility of $\pi$  and the resulting existence of a unitary $U\in \mfa$ such that ${\bc \pi(U)^* }\Om=V^* \Om_1$. 
{\bc This follows from the Kadison transitivity theorem  \cite[Theorem 10.2.1]{KR}.}

As for the last statement, suppose that $\ga$ is unitarily implementable in $\pi$, that is,
\beqa
\pi\circ \ga=\mrm{Ad}U_{\ga}\circ \pi
\eeqa
for some unitary $U_{\ga}$ on $\hil$. Thus we can write for any $A\in \mfa$
\begin{align}
\om(\ga(A))&=\lan \Om, \pi(\ga(A)) \Om\ran=\lan {\bc U_{\ga}^*}\Om, \pi(A) {\bc U_{\ga}^* }\Om\ran\\
&=\lan \pi(V_{\ga}^*)\Om, \pi(A) \pi(V_{\ga}^*)\Om\ran=\om(V_{\ga}AV_{\ga}^*),
\end{align}
where we used again \cite[Theorem 10.2.1]{KR} to find a unitary $V_{\ga}\in \mfa$ such that $\pi(V_{\ga}^*)\Om={\bc U_{\ga}^* }\Om$.

Now suppose that $\om$ and $\om\circ\ga$ are in the same sector, i.e., $\om=\om\circ\ga\circ i$ for some $i\in \mrm{In}(\mfa)$. Then  $\ga\circ i$
leaves $\om$ invariant{\bc,} hence it is unitarily implementable by the GNS theorem. As any $i\in \mrm{In}(\mfa)$ is unitarily implementable, we conclude the proof.  \qed
\section{Vector valued spherical harmonics}\label{appendix:vectorSH}
\setcounter{equation}{0}

 As we did not find a satisfactory reference, in this appendix we  summarize the  basic properties of the vector valued spherical harmonics from
Section~\ref{Map-T}:
\begin{equation}\label{vSH-one}
\pmb{Y}_{\ell m \pm} = \frac{1}{\sqrt{\ell(\ell +1)}}\pmb{a}_\pm Y_{\ell m} \quad \text{with} 
\quad \pmb{a}_+ = |\pmb{k}| \nabla_{\pmb{k}} \quad \text{and} \quad \pmb{a}_- = \hat{\pmb{k}} \times \pmb{a}_+,
\end{equation}
where $Y_{\ell m}$ are the usual spherical harmonics, orthonormal with respect to the measure $d\Om(\theta, \phi)=\sin\, \theta\, d\theta d\phi$.
For this purpose we recall that the total angular momentum of a photon is a self-adjoint operator on $L^2_{\mrm{tr}}(\real^3;\complex^3)$ given by
\beqa
 \vJ=\vL+\vS, \quad \vL:= -i \pmb{k} \times \nabla_{\pmb{k}}, \quad \vS=(S_1, S_2, S_3), \quad (S_k{\bc\pmb{\psi}} ):=i\pmb{e}_k\times {\bc \pmb{\psi}},
\eeqa
where ${\bc\pmb{\psi}}\in    L^2_{\mrm{tr}}(\real^3;\complex^3)$ and    $\{ \pmb{e}_k   \}_{k=1,2,3}$ is the canonical basis in $\real^3$.  
The operators ${\bc \{S_k\}_{k=1,2,3} }$  satisfy the standard angular momentum commutation relations and   $\vS^2=2$ holds true  (cf.~\cite[\S 58, Problem~2]{LL}).
\begin{proposition}
The  vector valued spherical harmonics $\pmb{Y}_{\ell m \pm}$, $\ell\in \nat$, $-\ell\leq m \leq \ell$,  given by  equation~\eqref{vSH-one}
\begin{enumerate}
\item[(1)] satisfy $\pmb{J}^2 \pmb{Y}_{\ell m \pm}= \ell(\ell+1)\pmb{Y}_{\ell m \pm}$ and $\pmb{J}_3 \pmb{Y}_{\ell m \pm}= m \pmb{Y}_{\ell m \pm}$;
\item[(2)] form (a) an orthonormal and (b) complete basis of $L^2_{\mathrm{tr}}(S^2; \mathbb{C}^3)$.  
\end{enumerate}
\end{proposition}

\begin{proof}
 \begin{enumerate}
  \item[(1)] Following \cite{BLP},   we compute
  on $C^2$ functions from $L^2_{\mathrm{tr}}(S^2; \mathbb{C}^3)$
   \begin{equation}\label{commLa}
  [L_i, a_{\pm p}] = i \epsilon_{ipq} a_{\pm q}, 
  \end{equation}
where $L_i$ and $a_p$ are the  Cartesian components of $\pmb{L},\pmb{a}$, respectively, and $\epsilon_{ip q}$ is the Levi-Civita symbol.
 Since $(S_i \pmb{\psi})_p = -i\epsilon_{ipq}\psi_q$, equation~\eqref{commLa} applied to $Y_{\ell m}$ (which are smooth functions  \cite[Chapter~IV]{Hobs}) can be written as
\begin{equation}
L_i a_{\pm p}Y_{\ell m} - a_{\pm p}L_i Y_{\ell m}= -(S_i \pmb{a}_\pm Y_{\ell m})_p
\end{equation}
which implies 
\begin{equation}\label{Ji}
(J_i \pmb{a}_\pm Y_{\ell m})_p = L_i a_{\pm p}Y_{\ell m} + (S_i \pmb{a}_\pm Y_{\ell m})_p = a_{\pm p} L_i Y_{\ell m}.
\end{equation}
Consequently, for $i= 3$ the above equation yields
\begin{equation}
J_3 (\pmb{a}_\pm Y_{\ell m}) = \pmb{a}_\pm L_3 Y_{\ell m}.
\end{equation}
Instead, by applying $J_i$ twice 
and summing over the index $i=1,2,3$, we find from \eqref{Ji},
\begin{equation}
\pmb{J}^2(\pmb{a}_\pm Y_{\ell m}) = \pmb{a}_\pm \pmb{L}^2 Y_{\ell m}. 
\end{equation}
Since the scalar spherical harmonics are eigenfunctions of the operators $\pmb{L}^2$ and $L_3$ with eigenvalues $\ell(\ell +1)$ and $m$, respectively, we arrive at (1).

\item[(2a)] 
We denote by $\nabla_t = \pmb{a}_+$ the gradient on the sphere. Then the $\pmb{Y}_{\ell m +}$ are orthogonal to the $\pmb{Y}_{\ell' m' -}$, which follows by Green's identity on $S^2$. Similarly for $\pmb{Y}_{\ell m +}$,
\begin{multline}
\frac{1}{\ell(\ell+1)}\int_{S^2} \nabla_t  \ov{Y_{\ell m}(\hat{\pmb{k}})}  \cdot \nabla_t Y_{\ell' m'}( \hat{\pmb{k}} ) \; d\Omega \\
= -\frac{1}{\ell(\ell +1)} \int_{S^2} \Delta_t   \ov{Y_{\ell m}( \hat{\pmb{k}} )} Y_{\ell' m'}( \hat{\pmb{k}} ) \; d\Omega = \delta_{\ell \ell'}\delta_{mm'},
\end{multline}
since $Y_{\ell m}$ are orthonormal eigenfunctions of $\Delta_t$ with eigenvalues $-\ell(\ell +1)$.
As for $\pmb{Y}_{\ell m -}$, we have
\begin{multline}
\int_{S^2} \big( \hat{\pmb{k}} \times \ov{\pmb{Y}_{\ell m +}(\hat{\pmb{k}})}  \big) \cdot \big( \hat{\pmb{k}} \times \pmb{Y}_{\ell' m' +}(\pmb{\hat{k}})  \big) \; d\Omega =\\
\int_{S^2} (\hat{\pmb{k}} \cdot \hat{\pmb{k}})\big(   \ov{\pmb{Y}_{\ell m +}(\hat{\pmb{k}}) }  \cdot \pmb{Y}_{ \ell' m' +}(\hat{\pmb{k}}) \big) - \big( \hat{\pmb{k}}\cdot \pmb{Y}_{\ell' m' +}(\hat{\pmb{k}}) \big) \big( \ov{\pmb{Y}_{\ell m+}(\hat{\pmb{k}})}\cdot \hat{\pmb{k}}\big)\; d\Omega = \delta_{\ell \ell'}\delta_{m m'},
\end{multline}
using
the above result for $\pmb{Y}_{\ell m +}$ and the fact that these are orthogonal to $\hat{\pmb{k}}$.

\item[(2b)] Completeness is a consequence of \cite[Theorem~3.4]{Wi}, which states that given a field of tangent vectors $\pmb{\alpha}$ in the class $C^3$ on the unit sphere $S^2$, there exist unique functions $F$ and $G$ of class $C^2$ on $S^2$ such that
\begin{equation}\label{stzero}
 \int_{S^2} F \, d\Omega = \int_{S^2} G\, d\Omega =0
\end{equation}
and 
\begin{equation}\label{stint}
\quad \pmb{\alpha} = \nabla_t F + \hat{\pmb{k}} \times \nabla_t G.
\end{equation}
We consider an arbitrary smooth field of tangent vectors $\pmb{\alpha}$ on $S^2$. As this is more restrictive than the hypothesis in Wilcox's theorem above, we can apply this theorem to $\pmb{\alpha}$, yielding functions $F$ and $G$ of class $C^2$ on $S^2$ with the properties \eqref{stzero} and \eqref{stint} above.

 Let us {\bc now} decompose {\bc $F$, $G$ into sums which converge in $L^2(S^2)$:}
\beqa
F=\sum_{\ell,m}c_{\ell m}Y_{\ell m}, \quad G=\sum_{\ell,m}d_{\ell m}Y_{\ell m},
\eeqa
{\bc and substitute them to equation (\ref{stint}). In order to exchange the sums with the action of $\nabla_t$, $(\hat \vk\times \nabla_t) $, we 
do the following computation for any  $\pmb{\al}_1$  in the domain of the adjoint maps $\nabla_t^{*}, \, (\hat \vk\times \nabla_t)^*$. 
(For example, we can choose $\pmb{\al}_1$ smooth. As we indicate below, such vector fields form a dense subspace in $L^2_{\tr}(S^2;\complex^3)$)}.
 \begin{align}
\lan \pmb{\al}_1, \pmb{\al}\ran&= \lan \pmb{\al}_1, \nabla_t F + \hat{\pmb{k}} \times \nabla_t G\ran\non\\
&= \lan  (\nabla_t)^* \pmb{\al}_1,  F\ran +   \lan (  \hat{\pmb{k}} \times \nabla_t)^*\pmb{\al}_1, G\ran\non\\
&=\sum_{\ell,m} c_{\ell,m} \lan  (\nabla_t)^* \pmb{\al}_1,  Y_{\ell m}\ran +  \sum_{\ell,m} d_{\ell,m} \lan (  \hat{\pmb{k}} \times \nabla_t)^*\pmb{\al}_1, Y_{\ell m}\ran\non\\
&=\sum_{\ell,m} c_{\ell,m} \lan   \pmb{\al}_1,  \nabla_tY_{\ell m}\ran +  \sum_{\ell,m} d_{\ell,m} \lan \pmb{\al}_1, (  \hat{\pmb{k}} \times \nabla_t)Y_{\ell m}\ran\non\\
&=\sum_{\ell,m} c_{\ell,m}\sqrt{\ell (\ell+1)} \lan   \pmb{\al}_1,  \pmb{Y}_{\ell m+}\ran +  \sum_{\ell,m} d_{\ell,m}\sqrt{\ell(\ell+1)} \lan \pmb{\al}_1,
 \pmb{Y}_{\ell m-}\ran. \label{completeness-computation}
\end{align}
To proceed, we need to show that
\beqa
\sum_{\ell,m} |c_{\ell,m}|^2 \ell (\ell+1)<\infty, \quad  \sum_{\ell,m} |d_{\ell,m}|^2 \ell (\ell+1)<\infty. \label{summability-bound}
\eeqa
To this end, we compute using formula~(\ref{stint})
\begin{align}
\lan \pmb{Y}_{\ell m +}, \pmb{\al}  \ran&=  \fr{1}{\sqrt{\ell(\ell+1)} } \lan \nabla_t Y_{\ell m}, \pmb{\al}\ran= 
\fr{1}{\sqrt{\ell(\ell+1)} } \lan \nabla_t Y_{\ell m}, \nabla_t F\ran\non\\
&=-\fr{1}{\sqrt{\ell(\ell+1)} }\lan \Delta_t Y_{\ell m},  F\ran=-\sqrt{\ell(\ell+1)} \lan Y_{\ell m},  F\ran.
\end{align}
Since  $c_{\ell,m}=\lan Y_{\ell m}, F\ran$  and $\pmb{Y}_{\ell m +}$ form an orthonormal system, the first bound in (\ref{summability-bound}) follows.
The second bound is proven analogously.

Given (\ref{summability-bound}),  computation (\ref{completeness-computation}) gives
\beqa
\lan \pmb{\al}_1, \pmb{\al}- \sum_{\ell m \pm} \ti{c}_{\ell m \pm} \pmb{Y}_{\ell m \pm}\ran=0 
\eeqa
for some square-summable coefficients $ \ti{c}_{\ell m \pm}$. By taking supremum over all $\pmb{\al}_1$ in
the dense domain specified above, subject to $\|\pmb{\al}_1\|\leq 1$, we get
\beqa
\|\pmb{\al}- \sum_{\ell m \pm} \ti{c}_{\ell m \pm} \pmb{Y}_{\ell m \pm}\|=0
\eeqa
which implies that any smooth $\pmb{\al}$ is in the closed subspace spanned by  $\pmb{Y}_{\ell m \pm}$.

Hence it only remains to be shown that the smooth vector fields on the unit sphere are $L^2$-dense in the space of all $L^2$ vector fields.
For this, we consider a generic $L^2$ vector field $\pmb{\beta}$ and we split it into a sum $\pmb{\beta}= \pmb{\beta}_n + \pmb{\beta}_s$, where $\pmb{\beta}_{n,s}$ have support in the north and  south hemisphere, respectively.
Now stereographic projections map each hemisphere to a circle in $\mathbb{R}^2$, and the transformed $\pmb{\beta}_{n,s}$  can be approximated by smooth vector fields on a slightly larger circle. Applying the inverse transformation yields the result. \qedhere

\end{enumerate}

 \end{proof}

\end{appendix}

\end{document}